\documentclass[]{llncs}
\usepackage[T1]{fontenc}
\usepackage[utf8]{inputenc}
\usepackage{lmodern}
\usepackage[english]{babel}
\usepackage{titletoc}
\usepackage{multicol}

\usepackage{amsmath,amssymb}
\usepackage{mathpartir}
\usepackage{listings}
\usepackage{xspace}
\usepackage[linktocpage,colorlinks=true,allcolors=blue]{hyperref}
\usepackage{mathtools}
\usepackage{multicol, makecell, multirow}
\usepackage{booktabs} 
\usepackage{geometry}
\usepackage[nottoc,notlot,notlof]{tocbibind}
\usepackage{pdfpages}
\usepackage{xcolor}
\usepackage{ulem}

\usepackage{tikz}
\usetikzlibrary{arrows,arrows.meta}
\usetikzlibrary{positioning,calc}
\usetikzlibrary{fit,backgrounds}
\usetikzlibrary{decorations.shapes}
\usetikzlibrary{shapes.symbols}
\tikzstyle{graph}=[>=stealth',semithick,inner sep=2pt,minimum width=6mm]
\tikzstyle{node}=[draw,circle,fill=white]

\tikzset{
    >=stealth',
}

\renewcommand{\emph}[1]{{\itshape #1}}

\makeatletter
\def\orcidID#1{\smash{\href{http://orcid.org/#1}{\protect\raisebox{-1.25pt}{\protect\includegraphics{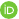}}}}}
\makeatother

\newcommand{\nb}[1]{\marginpar{\scriptsize #1}}
\newcommand\sidenote[1]{\nb{\textcolor{blue}{#1}}}

\newcommand{\todo}[1]{{\color{red} [TODO: #1]}}

\renewcommand{\todo}[1]{}
\renewcommand\sidenote[1]{}

\definecolor{extgreen}{rgb}{0.4,0.7,0.43}

\newcommand{\ext}[1]{{\color{extgreen} #1}}

\newcommand{\vamos}{\textsc{Vamos}\xspace}

\newcommand{\pe}{PE\xspace}
\newcommand{\mpe}{MPE\xspace}
\newcommand{\mpt}{MPT\xspace}
\newcommand{\dmpt}{DMPT\xspace}
\newcommand{\ksafety}{$k$-safety\xspace}

\definecolor{codegreen}{rgb}{0,0.6,0}
\definecolor{codegray}{rgb}{0.5,0.5,0.5}
\definecolor{codepurple}{rgb}{0.58,0,0.82}
\definecolor{backcolour}{rgb}{0.99,0.99,0.99}

\lstdefinestyle{mystyle}{
    backgroundcolor=\color{backcolour},
    commentstyle=\color{istagreen},
    keywordstyle=\color{istared},
    numberstyle=\footnotesize\color{codegray},
    stringstyle=\color{codepurple},
    basicstyle=\ttfamily\small,
    breakatwhitespace=false,
    breaklines=true,
    captionpos=b,
    keepspaces=true,
    numbers=left,
    numbersep=7pt,
    showspaces=false,
    showstringspaces=false,
    showtabs=false,
    tabsize=2,
    escapeinside={(@}{@)},
    frame=single
}

\definecolor{darkgreen}{RGB}{0,100,53}
\definecolor{istagreen}{RGB}{0,100,53}
\definecolor{istalightgreen}{RGB}{110,193,108} 
\definecolor{istalightgreen2}{RGB}{113,187,111}
\definecolor{istalightazure}{RGB}{218,238,239}
\definecolor{istablack}{RGB}{26,26,24}
\definecolor{istablue}{RGB}{37,59,144}
\definecolor{istaazure}{RGB}{0,154,163}
\definecolor{istalightorange}{RGB}{254,218,163}
\definecolor{istalightorange2}{RGB}{252,215,184}
\definecolor{istared}{RGB}{225,9,44}
\definecolor{istaorange}{RGB}{247,166,0}
\definecolor{istaorange2}{RGB}{244,152,0}

\usepackage[noend,linesnumbered,commentsnumbered,ruled]{algorithm2e}

\SetCommentSty{decmt}
\SetKwInput{KwData}{Input}
\SetKwInput{KwResult}{Output}
\SetKw{KwBreak}{break}
\SetKw{KwCont}{continue}
\SetKw{KwContinue}{continue}
\SetNlSty{textbf}{\color{black}}{}
\SetNlSkip{1.1em}

\makeatletter
\renewcommand{\algocf@Vsline}[1]{
  \strut\par\nointerlineskip
  \algocf@bblockcode%
  \algocf@push{\skiprule}
  \hbox{{\color{gray!20}\vrule width 0.01pt}
    \vtop{\algocf@push{\skiptext}
      \vtop{\algocf@addskiptotal #1}}}
  \algocf@pop{\skiprule}
  \algocf@eblockcode%
}
\makeatother

\newcommand{\substr}[3]{{#1}[{#2}..{#3}]}
\newcommand{\suff}[2]{{#1}[{#2}..]}

\newcommand{\iter}[2]{{{#1}^\circledast{#2}}}
\newcommand{\lab}[2]{{({#2})_{#1}}}
\newcommand{\labm}[2]{{[{#2}]_{#1}}}

\newcommand{\cond}[1]{\textcolor{blue}{[#1]}}
\newcommand{\teval}[2]{[\![#2]\!]_{#1}}
\newcommand{\val}[1]{\mathit{#1}}

\newcommand{\mconc}{\odot}
\newcommand{\mcomp}{\mconc}
\newcommand{\ML}{\mathbb{M}_L}
\newcommand{\step}[2]{\xRightarrow{#1,#2}}

\newcommand{\E}{\mathbb{E}}

\newcommand{\rname}[1]{\textsc{~(\textcolor{istablue}{#1})}}
\newcommand{\rnamet}[1]{\textsc{\textcolor{istablue}{#1}}\xspace}

\begin{document}

\title{Monitoring Hyperproperties With\\Prefix Transducers}
\titlerunning{Monitoring Hyperproperties With Prefix Transducers}
\author{Marek Chalupa\orcidID{0000-0003-1132-5516} \and
Thomas A. Henzinger\orcidID{0000-0002-2985-7724}
}
\authorrunning{M. Chalupa, T. A. Henzinger}
\institute{Institute of Science and Technology Austria (ISTA), Klosterneuburg, Austria}

\maketitle

\begin{abstract}

Hyperproperties are properties that relate multiple execution traces.
Previous work on monitoring hyperproperties focused on synchronous hyperproperties, usually specified in HyperLTL.
When monitoring synchronous hyperproperties, all traces are assumed to proceed at the same speed.
We introduce (multi-trace) \emph{prefix transducers} and show how to use them for monitoring synchronous as well as,
for the first time, asynchronous hyperproperties.
Prefix transducers map multiple input traces into one or more output traces
by incrementally matching prefixes of the input traces against expressions similar to regular expressions.
The prefixes of different traces which are consumed by a single matching step of the monitor may have different lengths.
The deterministic and executable nature of prefix transducers makes them more suitable as an intermediate formalism
for runtime verification than logical specifications,
which tend to be highly non-deterministic, especially in the case of asynchronous hyperproperties.
We report on a set of experiments about monitoring asynchronous version of
observational determinism.

\end{abstract}

\section{Introduction}

Hyperproperties~\cite{ClarksonS10} are properties that relate multiple execution traces of a system to each other.
One of the most prominent examples of hyperproperties nowadays are the information-flow security policies~\cite{McLean90}.
Runtime monitoring~\cite{Bartocci18} is a lightweight formal method for analyzing the behavior of a system by checking dynamic
execution traces against a specification.
For hyperproperties, a monitor must check relations between multiple traces.
While many hyperproperties cannot be monitored in general~\cite{AgrawalB16,BrettSB17,BonakdarpourSS18,FinkbeinerHST19,StuckiSSB19},
the monitoring of hyperproperties can still yield useful results, as we may detect their violations~\cite{FinkbeinerHST18}.

Previous work on monitoring hyperproperties focused on HyperLTL specifications~\cite{AgrawalB16,BrettSB17}, or other synchronous
hyperlogics~\cite{AcetoAAF22}.
Synchronous specifications model processes that progress at the same speed in lockstep, one event on each trace per step.
The synchronous time model has been found overly restrictive for specifying hyperproperties of asynchronous processes,
which may proceed at varying speeds~\cite{BartocciFHNC22,BaumeisterCBFS21,BozzelliPS21,GutsfeldMO21}.
A more general, asynchronous time model allows multiple traces to proceed at different speeds, independently of each other,
in order to wait for each other only at certain synchronization events.
As far as we know, there has been no previous work on the runtime monitoring of asynchronous hyperproperties.

The important class of $k$-safety hyperproperties~\cite{ClarksonS10,FinkbeinerHT19} can be monitored by processing $k$-tuples of traces~\cite{AgrawalB16,FinkbeinerHT19}.
In this work, we develop and evaluate a framework for monitoring \ksafety hyperproperties under both, synchronous and asynchronous time models.
For this purpose, we introduce \emph{(multi-trace) prefix transducers}, which map multiple (but fixed) input traces into one or more output traces
by incrementally matching prefixes of the input traces against expressions similar to regular expressions.
The prefixes of different traces which are consumed by a single matching step of the transducer may have different lengths, which allows to proceed on the traces asynchronously.
By instantiating prefix transducers for different combinations of input traces,
we can monitor $k$-safety hyperproperties instead of monitoring only a set of fixed input traces.
The deterministic and executable nature of prefix transducers gives rise to natural monitors. This is in contrast with monitors synthesized
from logical specifications which are often highly non-deterministic, especially in the case of asynchronous (hyper)properties.

\newcommand{\num}[1]{\textcolor{istaazure}{#1}}

\begin{figure}[t]
\begin{center}
\begin{tikzpicture}[
    nd/.style = {align=left},
    ev/.style = {align=left, rounded corners, scale=0.9,
                 fill=gray!20, minimum width=1cm}
  ]
\node[nd] at (-1.0, 0) { \strut\\$t_1:$\\\strut };
\node[ev] at (0,   0) { \strut \texttt{I(l, \num{1})}};
\node[ev] at (1.8, 0) { \strut \texttt{I(h, \num{1})}};
\node[ev] at (3.6, 0) { \strut \texttt{O(l, \num{1})}};
\node[ev] at (5.4, 0) { \strut \texttt{O(l, \num{1})}};

\node[nd] at (-1.0, -2em) { \strut\\$t_2:$\\\strut };
\node[ev] at (0,    -2em) { \strut \texttt{I(l, \num{1})}};
\node[ev] at (1.8,  -2em) { \strut \texttt{I(h, \num{2})}};
\node[ev] at (3.6,  -2em) { \strut \texttt{O(l, \num{1})}};
\node[ev] at (5.4,  -2em) { \strut \texttt{O(l, \num{1})}};


\node[nd] at (-1.0, -4em) { \strut\\$t_3:$\\\strut };
\node[ev] at (0,    -4em) { \strut \texttt{I(l, \num{1})}};
\node[ev] at (1.8,  -4em) { \strut \texttt{Dbg(\num{1})}};
\node[ev] at (3.6,  -4em) { \strut \texttt{I(h, \num{2})}};
\node[ev] at (5.4,  -4em) { \strut \texttt{O(l, \num{1})}};
\node[ev] at (7.2,  -4em) { \strut \texttt{O(l, \num{1})}};

\end{tikzpicture}
\caption{Traces of abstract events.
  The event \texttt{I(x, \num{v})} signals an input of value \num{v} into variable \texttt{x};
  and \texttt{O(x, \num{v})} signals an output of value \num{v} from variable \texttt{x}.
  The event \texttt{Dbg(\num{b})} indicates whether the debugging mode is turned on or off.}
\label{fig:traces}
\end{center}
\end{figure}
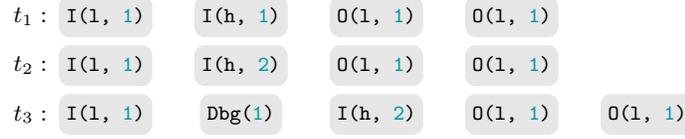

We illustrate prefix transducers on the classical example of \emph{observational determinism}~\cite{ZdancewicM03,HuismanWS06}.
Informally, observational determinism (OD) states that whenever two execution traces agree on \emph{low} (publicly visible) inputs,
they must agree also on low outputs,
thus not leaking any information about \emph{high} (secret) inputs.
Consider, for example, the traces $t_1$ and $t_2$ in Figure~\ref{fig:traces}.
These two traces satisfy OD, because they have the same low inputs (events \texttt{I(l, $\cdot$)})
and produce the same low outputs (events \texttt{O(l, $\cdot$)}).
All other events in the traces are irrelevant to OD.
The two traces even satisfy synchronous OD, as the input and output events appear at the same positions in both traces,
and thus they can be analysed by synchronous-time monitors
(such as those based on HyperLTL~\cite{AgrawalB16,BrettSB17,FinkbeinerHST18,PinisettySS18,FinkbeinerHST20}),
or by the following prefix transducer:

\begin{center}
  \begin{tikzpicture}[
    q/.style = {scale=1, draw, circle, minimum width=7mm, align=left},
    lab-edge/.style = {scale=0.8, align=left}
    ]
  \node[q] (0) at (0,0) {$q_{0}$};
  \node[q] (1) at (5,0) {$q_{1}$};
  \node[q] (2) at (-5,0) {$q_{2}$};

  \draw[->] ($(0.south) + (0, -3mm)$) to (0);
  \draw[->] (0) to[loop] node[lab-edge, yshift=10mm, xshift=9mm]
    {
      \begin{tabular}{l}
        $\tau_1: \iter{\_}{\lab{e_1}{I_l + O_l}}$\\
        $\tau_2: \iter{\_}{\lab{e_2}{I_l + O_l}}$\\
      \midrule
      \multicolumn{1}{c}{\cond{$\tau_1[e_1] = \tau_2[e_2]$}}
      \end{tabular}
      $\leadsto \tau_o \mapsto \top$
    }
    (0);
  \draw[->] (0) to node[lab-edge, yshift=10mm, xshift=18mm]
    {
      \begin{tabular}{l}
        $\tau_1: \iter{\_}{\lab{e_1}{O_l}}$\\
        $\tau_2: \iter{\_}{\lab{e_2}{O_l}}$\\
      \midrule
      \multicolumn{1}{c}{\cond{$\tau_1[e_1] \not = \tau_2[e_2]$}}
      \end{tabular}
      $\leadsto \tau_o \mapsto \bot$
    }
    (1);
  \draw[->] (0) to node[lab-edge, yshift=10mm, xshift=0mm]
    {
      \begin{tabular}{l}
        $\tau_1: \iter{\_}{\lab{e_1}{I_l}}$\\
        $\tau_2: \iter{\_}{\lab{e_2}{I_l}}$\\
      \midrule
      \multicolumn{1}{c}{\cond{$\tau_1[e_1] \not = \tau_2[e_2]$}}
      \end{tabular}
      $\leadsto \tau_o \mapsto \top$
    }
    (2);

\end{tikzpicture}
\end{center}

The transducer reads two traces $\tau_1$ and $\tau_2$ that are instantiated with actual traces, e.g.,~$t_1$ and $t_2$.
It starts in the initial state $q_0$ and either repeatedly takes the self-loop transition, or goes into one of the states $q_1$
or $q_2$ where it gets stuck.
The self-loop transition matches the shortest prefix of $\tau_1$ that contains any events until a low input or output is found.
This is represented by the \emph{prefix expression}
$\iter{\_}{{(I_l + O_l)}}$, where we use $I_l$ (resp.\ $O_l$) to represent any low input (resp.\ output) event,
and $\_$ stands for any event that does not match the right-hand side of $\iter{}{}$.
The same pattern is independently matched by this transition also against the prefix of $\tau_2$.
Moreover, the low input or output events found on traces $\tau_1$ and $\tau_2$ are labeled by $e_1$ and $e_2$, resp.
The self-loop transition is taken if the prefixes of $\tau_1$ and $\tau_2$ match the expressions and, additionally, the
condition
\textcolor{blue}{$\tau_1[e_1] = \tau_2[e_2]$} is fulfilled.
The term $\tau[e]$ denotes the sequence of events in trace $\tau$ on the position(s) labeled by $e$.
Therefore, the condition
\textcolor{blue}{$\tau_1[e_1] = \tau_2[e_2]$} asserts that the matched input
or output events must be the same (both in terms of type and values).
If the transducer takes the self-loop transition, it outputs (appends) the symbol $\top$ to the output trace $\tau_o$ (as stated by the right-hand side of $\leadsto$).
Then, the matched prefixes are consumed from the input traces and the transducer continues with matching the rest of the traces.
The other two edges are processed analogously.

It is not hard to see that the transducer decides if OD holds for two \emph{synchronous} input traces.
State $q_0$ represents the situation when OD holds but may still be violated in the future.
If the self-loop transition over $q_0$ cannot be taken, then (since we now assume synchronised traces), the matched prefixes must end either with different low input or different low output events.
In the first case, OD is satisfied by the two traces and the transducer goes to state $q_2$ where it gets stuck (we could, of course, make the transducer total to avoid getting stuck, but in this example we are interested only in its output).
In the second case, OD is violated and before the transducer changes the state to $q_1$,
it appends $\bot$ to the output trace $\tau_o$.
OD is satisfied if $\tau_o$ does not contain (end with) $\bot$ after finishing reading (or getting stuck on) the input traces.

The transducer above works also for monitoring asynchronous OD, where the low input and output events are misaligned by ``padding'' events (but it requires that there is the same number and order of low input and output events on the input traces -- they are just misaligned and possibly carry different values; the general setup where the traces can be arbitrary is discussed in Section~\ref{sec:experiments}).
The transducer works for asynchronous OD because the prefix expressions for $\tau_1$ and $\tau_2$ are matched independently, and thus they can match prefixes of different lengths.
For example, for $\tau_1 = t_1$ and $\tau_2 = t_3$, the run consumes the traces in the following steps:

\begin{center}
\begin{tikzpicture}[
    nd/.style = {align=left},
    ev/.style = {align=left, rounded corners, scale=0.9,
                 fill=gray!20, minimum width=1cm}
  ]
\node[nd] at (-1.2, 0) { \strut\\$t_1:$\\\strut };
\node[ev] at (0,   0) { \strut \texttt{I(l, \num{1})}};
\node[ev] at (3.6, 0) { \strut \texttt{I(h, \num{1})}};
\node[ev] at (5.4, 0) { \strut \texttt{O(l, \num{1})}};
\node[ev] at (7.2, 0) { \strut \texttt{O(l, \num{1})}};

\node[nd] at (-1.2, -2em) { \strut\\$t_3:$\\\strut };
\node[ev] at (0,    -2em) { \strut \texttt{I(l, \num{1})}};
\node[ev] at (1.8,  -2em) { \strut \texttt{Dbg(\num{1})}};
\node[ev] at (3.6,  -2em) { \strut \texttt{I(h, \num{2})}};
\node[ev] at (5.4,  -2em) { \strut \texttt{O(l, \num{1})}};
\node[ev] at (7.2,  -2em) { \strut \texttt{O(l, \num{1})}};

\draw[] (0.95, 2em) to (0.95, -4em);
\draw[] (6.30, 2em) to (6.30, -4em);
\draw[] (8.10, 2em) to (8.10, -4em);

\node[] at (0,   2.5em) { \strut \textit{step 1}};
\node[] at (3.6, 2.5em) { \strut \textit{step 2}};
\node[] at (7.2, 2.5em) { \strut \textit{step 3}};

\end{tikzpicture}
\end{center}

Hitherto, we have used the output of prefix transducers to decide OD for the given two traces, i.e.,~to perform the monitoring task.
We can also define a prefix transducer that, instead of monitoring OD for the two traces,
transforms the asynchronous traces $\tau_1$ and $\tau_2$ into
a pair of synchronous traces $\tau'_1$ and $\tau'_2$ by filtering out ``padding'' events:

\begin{center}
\begin{tikzpicture} 
  \node[draw,circle, minimum width=6mm] (0) at (0,0) {$q_0$};

  \draw[->] ($(0.north) + (0, 5mm)$) to (0);
  \draw[->] (0) to[loop left] node[yshift=0mm,align=left]
    {
      \begin{tabular}{l}
        $\tau_1: \iter{\_}{\lab{e_1}{I_l + O_l}}$\\
      \end{tabular}
      $\leadsto$
      \begin{tabular}{l}
        $\tau'_1 \mapsto \tau_1[e_1]$\\
      \end{tabular}
    } (0);
  \draw[->] (0) to[loop right] node[yshift=0mm,align=left]
    {
      \begin{tabular}{l}
        $\tau_2: \iter{\_}{\lab{e_2}{I_l + O_l}}$\\
      \end{tabular}
      $\leadsto$
      \begin{tabular}{l}
        $\tau'_2 \mapsto \tau_2[e_2]$\\
      \end{tabular}
    } (0);

\end{tikzpicture}
\end{center}

In this example, the transducer appends every event labeled by $e_i$
to the output trace $\tau'_i$, and so it filters out all events except low inputs and outputs.
It reads and filters both of the input traces independently of each other.
The output traces from the transducer can then be forwarded to, e.g., a HyperLTL monitor\footnote{
 One more step is needed before we can use a HyperLTL monitor, namely, to transform
 the trace of abstract events from the example into a trace of sets of atomic propositions.
 This can be also done by a prefix transducer.
}.
\ext{\todo{compare to monitor automata (and maybe other automata models) -- in a new section?}}


\subsubsection{Contributions} This paper makes the following contributions:

\begin{itemize}
  \item We introduce multi-trace prefix expressions and transducers (Section~\ref{sec:pe} and Section~\ref{sec:mpt}).
    These are formalisms that can efficiently and incrementally process words (traces) either fully synchronously, or asynchronously with synchronization points.
  \item We suggest that prefix transducers are a natural formalism for specifying
    many synchronous and asynchronous $k$-safety hyperproperties,
    such as observational determinism.
  \item We design an algorithm for monitoring synchronous and asynchronous $k$-safety hyperproperties using prefix transducers (Section~\ref{sec:hyper}).
  \item We provide some experiments to show how our monitors perform (Section~\ref{sec:experiments}).
\end{itemize}

 \label{sec:intro}

\sidenote{
  "Variants of regular expressions have been used for a long time... special thing about ours expressions is the shortest prefix thing"
  - formal definition of k-safety property
  and say what all combinations of ...
  - say all the additional work that can be done
  - say that we could discard some traces or use some search strategy
}

\section{Prefix expressions}\label{sec:pe}
In this section, we define \emph{prefix expressions} -- a formalism
similar to regular expressions
designed to deterministically and unambiguously match prefixes of words.

\subsection{Preliminaries}

We model sequences of events as \emph{words} over finite non-empty alphabets.
\sidenote{In some examples it is not clear that the alphabet is finite}%
Given an alphabet $\Sigma$, the set of finite words over this alphabet is denoted as $\Sigma^*$. 
For two words $u = u_0...u_l\in\Sigma^*_1$ and $v=v_0...v_m\in\Sigma^*_2$,
their concatenation $u\cdot v$, also written $uv$ if there is no confusion,
is the word $u_0...u_lv_0...v_m \in (\Sigma_1 \cup \Sigma_2)^*$.
If $w = uv$, we say that $u$ is a prefix of $w$, written $u \le w$, and $v$
is a suffix of $w$. If $u \le w$ and $u \not = w$, we say that $u$ is a 
proper prefix of $w$.

For a word $w = w_0...w_{k-1} \in \Sigma^*$, we denote $|w| = k$
its length, $w[i] = w_i$ for $0 \le i < k$ its $i$-th element,
$\substr{w}{s}{e} = w_sw_{s+1}...w_e$ the sub-word beginning at index $s$
and ending at index $e$, and $\suff{w}{s} = w_sw_{s+1}...w_{k-1}$ its suffix
starting at index $s$.

Given a function $f: A \rightarrow B$, we denote $Dom(f) = A$ its domain.
Partial functions with domain $A$ and codomain $B$ are written as $A \hookrightarrow B$.
Functions with a small domain are sometimes given extensionally by listing
the mapping, e.g., $\{x \mapsto 1, y \mapsto 2\}$.
Given a function $f$, $f[x \mapsto c]$ is the function that coincides with $f$ on all elements except on $x$ where it is $c$.

\subsection{Syntax of prefix expressions}

Let $L$ be a non-empty set of labels (names) and $\Sigma$ a finite non-empty alphabet.
The syntax of \emph{prefix expressions (\pe)} is defined by the
following grammar:

\begin{equation*}
\begin{split}
  \alpha ::= &~\epsilon \mid a\mid (\alpha.\alpha) \mid (\alpha + \alpha) \mid (\iter{\alpha}{\beta})
               \mid \lab{l}{\alpha}\\
  \beta ::=  &~a\mid (\beta + \beta) \mid \lab{l}{\beta}
\end{split}
\end{equation*}

\noindent
where $a \in \Sigma$ and $l \in L$.
Many parenthesis can be elided if we let '$\iter{}{}$' (iteration) take precedence before
'.' (concatenation), which takes precedence before '+' (disjunction, plus).
We write $a . b$ as $ab$ where there is no confusion.
In the rest of the paper, we assume that a set of labels $L$ is implicitly given,
and that $L$ always has ,,enough'' labels.
We denote the set of all prefix expressions over the alphabet $\Sigma$
(and any set of labels $L$) as $\pe(\Sigma)$, and  $\pe(\Sigma, L)$ if we want to
stress that the expressions use labels from $L$.

The semantics of {\pe}s (defined later) is similar to the semantics of regular
expressions with the difference that a \pe is not matched against a whole word
but it matches only its prefix, and this prefix is the shortest one possible (and non-empty -- if not explicitly specified).
For this reason, we do not use the classical Kleene's iteration as it would
introduce ambiguity in matching.
For instance, the regular expression $a^*$ matches all the prefixes
of the word $aaa$. And even if we specify that we should pick the shortest one,
the result would be $\epsilon$, which is usually not desirable, because that means no progress in reading the word.
Picking the shortest non-empty prefix would be a reasonable solution in many cases,
but the problem with ambiguity persists. For example, the regular expression
$(ab)^*(a + b)^*$ matches the word $ab$ in two different ways, which introduces
non-determinism in the process of associating the labels with the matched positions.

To avoid the problems with Kleene's iteration, we use a binary iteration operator
that is similar to the \emph{until} operator in LTL in that it requires
some letter to appear eventually.
The expression $\iter{\alpha}{\beta}$ could be roughly defined as
$\beta + {\alpha}{\beta} + {\alpha}^2{\beta} +
\dots$ where we evaluate it from left to right and
$\beta$ must match \emph{exactly} one letter.
The restriction on $\beta$ is important to tackle ambiguity,
but it also helps efficiently evaluate the expression -- it is enough to look at
a single letter to decide whether to continue matching whatever follows
the iteration, or whether to match the left-hand side of the expression.
Allowing $\beta$ to match a sequence of letters would require
a look-ahead or backtracking and it is a subject of future extensions.
With our iteration operator, expressions like $\iter{(ab)}{\iter{(a + b)}{}}$
and $\iter{a}{}$ are malformed and forbidden already on the level of syntax.

Sub-expressions of a \pe can be labeled and the matching procedure
described later returns a list of positions in the word that were matched
for each of the labels.
We assume that every label is used maximally once, that is, no two
sub-expressions have the same label.
Labels in {\pe}s are useful for identifying the sub-word that matched particular
sub-expressions, which will be important in the next
section when we use logical formulae that relate sub-words from
\emph{different} words (traces).
Two examples of {\pe}s and their informal evaluation are:
\begin{itemize}

  \item $\iter{\lab{l}{a+b}}{a}$ -- match $a$ or $b$, associating them to
  $l$, until you see $a$. Because whenever the word contains $a$,
  the right-hand side of the iteration matches,
  the left part of the iteration never matches $a$ and $a$ is redundant in the
  left-hand side sub-expression.
  For the word $bbbaba$, the expression matches the prefix $bbba$
  and $l$ is associated with the list of position ranges $(0,0), (1,1), (2,2)$ that
  corresponds to the positions of $b$ that were matched by the sub-expression
  $(a+b)$.

  \item $\lab{l_1}{\iter{a}{b}}\lab{l_2}{\iter{(b+c)}{(a+d)}}$ -- match $a$
  until $b$ is met and call this part $l_1$; then match $b$ or $c$ until $a$ or
  $d$ and call that part $l_2$. For the word $aabbbada$, the expression
  matches the prefix $aabbba$. The label $l_1$ is associated with the range of
  positions $(0,2)$ containing the sub-word $aab$,
  and $l_2$ with the range~$(3,5)$ containing the sub-word $bba$.

  \end{itemize}

\subsection{Semantics of prefix expressions}

We first define \emph{m-strings} before we get to the formal semantics of {\pe}s.
An m-string is a sequence of pairs of numbers or a special symbol $\bot$.
Intuitively, $(p, \bot)$ represents the beginning of a match at position $p$ in
the analyzed word, $(\bot, p)$ the end of a match at position $p$, and
$(s, e)$ is a match on positions from $s$ to $e$.
The concatenation of m-strings reflects opening and closing the matches.

\begin{definition}[M-strings]
\emph{M-strings} are words over the alphabet
$M=(\mathbb{N} \cup \{\bot\})\times(\mathbb{N} \cup \{\bot\})$
with the partial concatenation function
$\mconc: M^* \times M \hookrightarrow M^*$ defined as

\[
\alpha\mconc(c, d) =
\begin{cases}
(c,d)             & \text{if } \alpha = \epsilon \land c \not = \bot \\
\alpha\cdot(c,d)  & \text{if } \alpha = \alpha'\cdot(a, b) \land b \not = \bot\\
\alpha'\cdot(a,d) & \text{if } \alpha = \alpha'\cdot(a, b) \land b = \bot \land c = \bot\\
\alpha'\cdot(c,d) & \text{if } \alpha = \alpha'\cdot(a, b) \land b = \bot \land c \not = \bot
\end{cases}
\]
\end{definition}

Every m-string is built up
from the empty string $\epsilon$ by repeatedly concatenating $(p, \bot)$ or
$(\bot, p)$ or $(p, p)$.
The concatenation $\mconc$ is only a partial function (e.g., $\epsilon \mconc
(\bot, \bot)$ is undefined), but we will use only the defined fragment of the
domain.
It works as standard concatenation if the last match was closed (e.g.,~$(0,
4)\mconc(7, \bot) = (0, 4)(7, \bot)$), but it overwrites the last match if we
start a new match without closing the old one (e.g.,~$(0, \bot)\mconc(2, \bot)
= (2, \bot)$). Overwriting the last match in this situation is valid, because
labels are assumed to be unique and opening a new match before the last match
was closed means that the old match failed.

We extend $\mconc$ to work with m-strings on the right-hand side:
\[
\alpha \mconc w =
\begin{cases}
  \alpha & \text{if } w = \epsilon\\
  (\alpha \mconc x) \mconc w' & \text{if } w = x \cdot w'
\end{cases}
\]

While evaluating a \pe, we keep one m-string per label used in the \pe. To do so
we use \emph{m-maps}, partial mappings from labels to m-strings.

\begin{definition}[M-map]
  Let $L$ be a set of labels and
  $M=(\mathbb{N} \cup \{\bot\})\times(\mathbb{N} \cup \{\bot\})$
  be the alphabet of m-strings. An \emph{m-map} is a partial function
  $m: L \hookrightarrow M^*$. Given two
  m-maps $m_1, m_2: L \hookrightarrow M^*$, we define their concatenation
  $m_1\mcomp m_2$ by
\[
(m_1\mcomp m_2)(l) = \sigma_1 \mconc \sigma_2
\]
for all $l\in Dom(m_1) \cup Dom(m_2)$ where $\sigma_i = \epsilon$ if $m_i(l)$
is not defined and $\sigma_i = m_i(l)$ otherwise.
\end{definition}

\noindent
We denote the set of all
m-maps $L \hookrightarrow M^*$ for the set of labels $L$ as $\ML$.

\smallskip
The evaluation of a \pe over a word $w$ is defined as an iterative application
of a set of rewriting rules that are inspired by language
derivatives~\cite{Brzozowski64,Antimirov96}. For the purposes of evaluation, we
introduce two new prefix expressions: $\bot$%
\footnote{Because {\pe}s and m-strings never interact together,
we use the symbol $\bot$ in both, but formally they are different symbols.}
and $\labm{l}{\alpha}$. \pe $\bot$ represents a failed match and
$\labm{l}{\alpha}$ is a ongoing match of a labeled sub-expressions
$\lab{l}{\alpha}$. Further, for a \pe $\alpha$, we define inductively that
$\epsilon \in \alpha$ iff
\begin{itemize}
  \item[-] $\alpha = \epsilon$, or
  \item[-] $\alpha$ is one of $(\alpha_0 + \alpha_1)$ or
           $\lab{l}{\alpha_0 + \alpha_1}$ or
           $\labm{l}{\alpha_0 + \alpha_1}$ and it holds that
$(\epsilon \in  \alpha_0 \lor \epsilon \in \alpha_1)$.
\end{itemize}

For the rest of this section, let us fix a set of labels $L$, an
alphabet $\Sigma$, and denote $\E = \pe(\Sigma, L)$ the set of all prefix
expressions over this alphabet and this set of labels.
The core of evaluation of {\pe}s is the \emph{one-step
relation} $\step{a}{p} \subseteq (\E \times \ML) \times (\E \times
\ML)$ defined for each letter $a$ and natural number $p$,
whose defining rules are depicted in Figure~\ref{fig:one-step-rules}.
We assume that the rules are evaluated modulo the equalities $\epsilon\cdot\alpha = \alpha\cdot\epsilon = \alpha$,
and we say that $\alpha'$
is a \emph{derivation} of $\alpha$ if $\alpha \step{a}{p} \alpha'$ for some $a$ and
$p$.

\begin{figure}[th!]
{
\begin{align*}
  \inferrule{~}{(\epsilon, M) \step{a}{p} (\bot, \emptyset)}\rname{Eps}&&
\inferrule{~}{(\bot, M) \step{a}{p} (\bot, \emptyset)}\rname{Bot}&&
\inferrule{~}{(a, M) \step{a}{p} (\epsilon, M)}\rname{Ltr}
\end{align*}
\begin{align*}
\inferrule{a \not = b}{(a, M) \step{b}{p} (\bot, \emptyset)}\rname{Ltr-fail}&&
\inferrule{(\alpha, M) \step{a}{p} (\alpha', M')\\
            \alpha' \not = \bot \\ \epsilon\not\in\alpha'}
{(\alpha\beta, M) \step{a}{p} (\alpha'\beta, M')}\rname{Concat}&&
\end{align*}
\begin{align*}
\inferrule{(\alpha, M) \step{a}{p} (\alpha', M')\\
            \epsilon\in\alpha' }
{(\alpha\beta, M) \step{a}{p} (\beta, M')}\rname{Concat-$\epsilon$}&&
\inferrule{(\alpha, M) \step{a}{p} (\alpha', M')\\
            \alpha' = \bot}
          {(\alpha\beta, M) \step{a}{p} (\bot, \emptyset)}\rname{Concat-$\bot$}
\end{align*}
\begin{align*}
\inferrule{(\alpha_0, M) \step{a}{p} (\alpha'_0, M'_0)\\
           (\alpha_1, M) \step{a}{p} (\alpha'_1, M'_1)\\
           \epsilon \in\alpha'_0 \lor \epsilon\in\alpha'_1}
          {(\alpha_0 + \alpha_1, M) \step{a}{p} (\epsilon, M_0'\mcomp M_1')}\rname{OR-end}
\end{align*}
\begin{align*}
\inferrule{(\alpha_0, M) \step{a}{p} (\alpha'_0, M'_0)\\
           (\alpha_1, M) \step{a}{p} (\alpha'_1, M'_1)\\
           \epsilon \not\in\alpha'_0 \land \epsilon\not\in\alpha'_1 }
           {(\alpha_0 + \alpha_1, M) \step{a}{p} (\alpha'_0 + \alpha'_1, M_0'\mcomp M_1')\\ \alpha'_0 + \alpha'_1 \text{ is reduced w.r.t } \alpha + \bot = \bot + \alpha =\alpha}\rname{OR}
\end{align*}
\begin{align*}
\inferrule{(\beta, M)\step{a}{p}(\beta', M')\\\epsilon \in \beta'}
          {(\iter{\alpha}{\beta}, M) \step{a}{p} (\epsilon, M')}\rname{Iter-end}&&
\inferrule{(\beta, M)\step{a}{p}(\beta', \_)\\
            \epsilon \not \in \beta'\\
            (\alpha, M)\step{a}{p}(\alpha', M')}
          {(\iter{\alpha}{\beta}, M) \step{a}{p} (\alpha'\iter{\alpha}{\beta}, M')}\rname{Iter}
\end{align*}
\begin{align*}
\inferrule{(\alpha, M)\step{a}{p}(\alpha', M')\\ \epsilon \not \in \alpha'}
{(\lab{l}{\alpha}, M) \step{a}{p} (\labm{l}{\alpha'}, \{l\mapsto(p, \bot)\}\mcomp M')}\rname{L-start}&&
\inferrule{(\alpha, M)\step{a}{p}(\alpha', M')\\ \epsilon \not \in \alpha'}
          {(\labm{l}{\alpha}, M) \step{a}{p} (\labm{l}{\alpha'}, M')}\rname{L-cont}&&
\end{align*}
\begin{align*}
\inferrule{(\alpha, M)\step{a}{p}(\alpha', M')\\ \epsilon \in \alpha'}
          {(\lab{l}{\alpha}, M) \step{a}{p} (\epsilon, M'\mcomp \{l\mapsto(p, p)\})}\rname{L-ltr}&&
\inferrule{(\alpha, M)\step{a}{p}(\alpha', M')\\ \epsilon \in \alpha'}
{(\labm{l}{\alpha}, M) \step{a}{p} (\epsilon, M'\mcomp\{l\mapsto(\bot, p)\})}\rname{L-end}
\end{align*}\\
}
\caption{One-step relation for evaluating prefix expressions.
  The rules are evaluated modulo the equalities $\epsilon\cdot\alpha = \alpha\cdot\epsilon = \alpha$.
}
\label{fig:one-step-rules}
\end{figure}

The first seven rules
in Figure~\ref{fig:one-step-rules} are rather standard. Rules for disjunction
are non-standard in the way that whenever an operand of a disjunction is
evaluated to $\epsilon$, the whole disjunction is evaluated to $\epsilon$ (rule
\rnamet{Or-end}) in order to obtain the shortest match.
Also, after rewriting a disjunction, operands that
evaluated to $\bot$ are dropped (unless the last one if we should drop all
operands).
The only non-shortening rule is \rnamet{Iter} that unrolls $\alpha$ if $\beta$ was not matched in $\iter{\alpha}{\beta}$.
Thus, evaluating $\iter{\alpha}{\beta}$ can first prolong the expression and then eventually end up again in $\iter{\alpha}{\beta}$ after a finite number of steps.
This does not introduce any problems as the set of derivations remains bounded~\cite{extended}.

There are four rules for handling labellings. 
Rule \rnamet{L-ltr} handles one-letter matches.
Rule \rnamet{L-start} handles the beginning of a match where the expression $\lab{l}{\alpha}$
gets rewritten to $\labm{l}{\alpha}$ so that we know we are currently matching
label $l$ in the next steps. Rules \rnamet{L-cont} and \rnamet{L-end} continue
and finish the match once the labeled expression evaluates to $\epsilon$.
Concatenating m-maps in the rules works well because of the assumption that no two
sub-expressions have the same label. Therefore, there are no collisions and we
always concatenate information from processing the same sub-expression.

The one-step relation is deterministic, i.e.,~there is always at
most one rule that can be applied and thus every \pe has a unique single derivation for a fixed letter $a$ and position $p$.

\begin{theorem}[Determinism of $\step{a}{p}$] \label{thm:pe:deterministic}
  For an arbitrary \pe $\alpha$ over alphabet $\Sigma$, and any $a\in\Sigma$
  and $p\in\mathbb{N}$, there exist at most one $\alpha'$ such that
  $\alpha \step{a}{p} \alpha'$ (that is, there is at most one defining rule
  of $\step{a}{p}$ that can be applied to $\alpha$).
\end{theorem}
\begin{proof}[Sketch]
  Multiple rules could be applied only if they match the same structure of $\alpha$
  (e.g., that $\alpha$ is a disjunction). But for such rules, the premises are pairwise
  unsatisfiable together.
\end{proof}

Having the deterministic one-step relation, we can define the
\emph{evaluation function} on words $\delta: \E \times {\Sigma^*} \rightarrow \E
\times \ML$ that returns the rewritten \pe and the m-map
resulting from evaluating the one-step relation for each single letter of the word:
$\delta(\alpha, w) = \overline\delta(\alpha, w, 0, \emptyset)$ where
\vspace*{-1mm}
\[
\overline\delta(\alpha, w, p, M) =
\begin{cases}
  (\alpha, M)           & \text{if } w = \epsilon\\
   \overline\delta(\alpha', w', p+1, M')
                           & \text{if } w = aw' \land (\alpha, M)\step{a}{p}(\alpha', M')\\
\end{cases}
\]



We also define the \emph{decomposition function}
$\rho: \E \times \Sigma^* \rightarrow (\Sigma^* \times \ML \times \Sigma^*) \cup \{\bot\}$
that decomposes a word $w \in \Sigma^*$
into the matched prefix with the resulting m-map, and the rest of $w$:
\[
\rho(\alpha, w) =
\begin{cases}
  (u, m, v) & \text{if }~w = uv \land \delta(\alpha, u) = (\epsilon, m)\\
  \bot & \text{otherwise}
\end{cases}
\]

Function $\rho$ is well-defined as there is at most one $u$ for which
$\delta(\alpha, u) = (\epsilon, \_)$. This follows
from the determinism of the one-step relation.
Function $\rho$ is going to be important in the next section.

There are only finitely many {\pe}s that we can obtain
from a \pe by repeated application of $\step{a}{p}$ for
any $a$ and a fixed $p$.

\begin{theorem}[Finite derivations of {\pe}s]\label{thm:pe:finite_derivs}
  Given a number $p\in\mathbb{N}$ and an arbitrary \pe $\alpha$ over an alphabet $\Sigma$, the set
  \[
  D_p(\alpha) = \{\alpha' \mid \alpha \step{a_0}{p} ... \step{a_k}{p} \alpha', a_1, ..., a_k \in \Sigma\}
  \]
  is finite.
\end{theorem}
\begin{proof}[Sketch]
  Every rule except \textsc{Iter} is shortening.
  Application of $\step{a}{p}$ on $\iter{\alpha}{\beta}$ is either $\epsilon$
  or $\alpha'\cdot\iter{\alpha}{\beta}$. If $\alpha$ contains nested iteration,
  the expression can get extended further to the left,
  but because the nesting of iterations must be finite,
  at some point the left-hand side of the left-most iteration,
  let us call this iteration expression $\xi$,
  has no nested iteration and repeatedly applying one-step
  relation leads to shortening the expression until $\xi$ again becomes the left-most
  expression (or the whole expression evaluates to $\bot$).
  From this point on,
  any expression obtained by one-step relation is either shorter than the original one
  or it has been already seen.
\end{proof}

To strengthen the intuition behind the proof sketch above, here are depicted
possible evaluations of an expression over the alphabet $\{a, b, c\}$ with two nested iterations (we do not show the parameter~$p$):

\begin{center}
\begin{tikzpicture}
  \node[anchor=west](1) at (0, 0) {$\iter{(\iter{ab}{c})}{c}$};
  \node[anchor=west](2) at (3, 0) {$b(\iter{ab}{c})\iter{(\iter{ab}{c})}{c}$};
  \node[anchor=west](3) at (7, 0) {$(\iter{ab}{c})\iter{(\iter{ab}{c})}{c}$};
  \node[below=of 1](e) {$\epsilon$};
  \node[below=of 2](b) {$\bot$};
  \draw (1) edge[double, ->] node[yshift=3mm] {$a$}(2);
  \draw (1) edge[double, ->] node[xshift=3mm] {$c$}(e);
  \draw (2) edge[double, ->] node[yshift=3mm] {$b$}(3);
  \draw (3) edge[bend left, double, ->] node[yshift=-3mm] {$a$}(2);
  \draw (3) edge[bend right, double, ->] node[yshift=3mm] {$c$}(1);

  \draw (1) edge[double, ->] node[yshift=3mm] {$b$}(b);
  \draw (2) edge[double, ->] node[xshift=4mm] {$a,c$}(b);
  \draw (3) edge[double, ->, bend left] node[yshift=-3mm] {$b$}(b);
\end{tikzpicture}
\end{center}

Thanks to theorems~\ref{thm:pe:deterministic} and \ref{thm:pe:finite_derivs},
the evaluation function $\delta$ for a fixed \pe $\alpha$ over an alphabet $\Sigma$
can be translated into a unique \emph{prefix-free}~\cite{HanW04} finite state transducer~\cite{VeanesHLMB12}.
The transducer has for states the set $D_p(\alpha)$, and edges are formed as in the example above (notice that the expressions in the example above are exactly the expressions from $D_p(\iter{(\iter{ab}{c})}{c})$), i.e.,~for
any $a\in\Sigma$, there is the transition $\alpha_0 \xrightarrow{a,p/M} \alpha_1$
between two states iff $(\alpha_0, \emptyset) \step{a}{p} (\alpha_1, M)$. The
initial state of a transducer is the state $\alpha$
and the accepting state is $\epsilon$.
A run of the transducer on a word $w = w_0w_1...w_k$ is a sequence of
adjacent edges starting from the initial state such that the guard on $i$-th edge is $w_i, i$:
$$
(\alpha)\xrightarrow{w_0, 0/M_0} (\alpha_1) \xrightarrow{w_1,1/M_1} ...
\xrightarrow{w_{k-1},k-1/M_{k-1}}(\alpha_k)
$$
Further, it holds that $\delta(\alpha, w) = (\alpha_k,
M_0\mcomp...\mcomp M_{k-1})$.
An example of a transducer for
$b(\iter{\lab{l}{ab}}{b})(b+a)$ over the alphabet $\{a, b\}$
is shown in Figure~\ref{fig:delta_transducer}.

\begin{figure}
\begin{center}
\begin{tabular}{c}
\begin{tikzpicture}[
  nd/.style = {fill=gray!5,draw, rounded corners, inner sep=3pt}
]
\node[nd] (0) at (-1, 0)  {\color{black}$b(\iter{\lab{l}{ab}}{b})(b+a)$};
\node[nd] (1) at (3, 0)   {\color{black}$(\iter{\lab{l}{ab}}{b})(b+a)$};
\node[nd] (2) at (7, 0)   {\color{black}$(b+a)$};
\node[nd] (3) at (3, -1.5){\color{black}$[b]_l(\iter{\lab{l}{ab}}{b})(b+a)$};
\node[nd,double, color=istagreen,fill=white]
           (5) at (7, -1.5) {\color{istagreen}$\epsilon$};

\draw[->] ($(0.west) + (-5mm, 0)$)  to (0);
\draw[->] (0) edge node[yshift=3mm]{$b,p/\emptyset$} (1);
\draw[->] (1) edge node[yshift=3mm]{$b,p/\emptyset$} (2);
\draw[->] (1) edge node[ xshift=-13mm]{$a,p/\{l\mapsto(p, \bot)\}$} (3);
\draw[] ($(3.north) + (5mm, 0)$)
          edge[->] node[xshift=13mm]{$b,p/\{l \mapsto (\bot, p) \}$}
        ($(1.south) + (5mm, 0)$);
\draw[->] (2) to node[xshift=4mm]{$*/\emptyset$} (5);

\end{tikzpicture}
\end{tabular}
\end{center}
\caption{Finite transducer for \pe $b(\iter{\lab{l}{ab}}{b})(b+a)$.
         Redundant states and transitions
         were omitted.}
\label{fig:delta_transducer}
\end{figure}
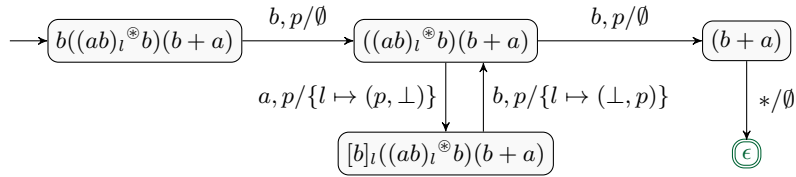

Strictly formally, this transducer as shown in Figure~\ref{fig:delta_transducer}
would be \emph{symbolic} transducer~\cite{VeanesHLMB12} as the position $p$ on edges represents any natural number.
However, $p$ is rather a parameter of the transducer and it only propagates into the output of some edges.
Therefore, it would be more precise to avoid having $p$ in the guard of the edge
and, instead, output functions that take $p$ as the only parameter. For example,
instead of the m-map $\{l \mapsto (p, \bot)\}$,
we would output a map that returns the m-map: $\lambda p.\{l \mapsto (p, \bot)\}$.
The semantics of the output of the transducer then would be modified
to supply and evaluate the outputs on the right sequence of $p$s.
This approach, although formally better, seems unnecessarily complicated
for our purposes.

The fact that {\pe}s correspond to \emph{prefix-free} transducers
suggests how to compose and perform other operations with {\pe}s.

\section{Multi-trace prefix expressions and transducers}\label{sec:mpt}

In this section, we define \emph{multi-trace prefix expressions} and
\emph{multi-trace prefix transducers}.

\subsection{Multi-trace prefix expressions}
A multi-trace prefix expression (\mpe)
matches prefixes of multiple words.
Every \mpe consists of prefix expressions associated to input words and
a \emph{condition} which is a logical formula that must be satisfied by the matched prefixes.

\begin{definition}[Multi-trace prefix expression]
\emph{Multi-trace prefix expression (\mpe)}
over trace variables $V_\tau$ and alphabet $\Sigma$ is a list of pairs together with a formula:
$$(\tau_0, e_0), \ldots, (\tau_k, e_k)[\varphi]$$
where $\tau_i \in V_\tau$ are trace variables and $e_i$ are {\pe}s over the alphabet
$\Sigma$. The formula $\varphi$ is called the \emph{condition} of \mpe.
We require that for all $i \not = j$, labels in $e_i$ and $e_j$ are distinct.
\end{definition}

If the space allows, we typeset {\mpe}s over multiple lines as can be seen
in the examples of prefix transducers throughout the paper.

\mpe conditions are logic formulae over m-strings and input words\footnote{
The conditions can be almost arbitrary formulae,
depending on how much we are willing to pay for their evaluation.
They could even contain nested {\mpe}s. As we will see later, \mpe conditions are evaluated only after matching the prefixes which gives us a lot of freedom in choosing the logic. For the purposes of this work,
however, we use only simple conditions that we need in our examples
and evaluation.
}.
Let $V_\tau$ be a set of trace variables and $L$ a set of labels.
Atomic terms of \mpe conditions are $l$ for each $l \in L$,
$\tau[l]$ for any label $l$ and trace variable $\tau \in V_\tau$,
and constants (words over arbitrary alphabet and m-strings).
%
A well-formed \mpe condition is a boolean formula built up inductively from terms with these rules: 
\begin{itemize}
  \item $t = t'$ is a well-formed condition for $t, t'$ atomic terms
  \item if $C$ is a well-formed condition, $\neg C$ is a well-formed condition
  \item $(C \land C')$ are well-formed conditions for well-formed conditions $C, C'$
\end{itemize}

As usual, we can omit parenthesis using the precedence rules that the equality have
the highest priority, then negation ($\neg$) and then conjunction
($\land$). We write $t \not = t'$ for $\neg(t = t')$ and we use the classical
shortcuts for disjunction and implication: $C \lor C' \equiv \neg(\neg C \land \neg C')$,
and $C \implies C' \equiv \neg C \lor C'$.

For an assignment $\sigma: V_\tau \rightarrow \Sigma^*$ of trace variables to words
and an m-map $M$, we define the evaluation of atomic terms $\teval{\sigma, M}{\cdot}$
in the following way:

\begin{align*}
  \teval{\sigma, M}{c} =&~ c \text{ for any constant } c\\
  \teval{\sigma, M}{l} =&~
  \begin{cases}
    M[l]       & \text{if } l \in Dom(M)\\
    \epsilon   & \text{otherwise}
  \end{cases}\\
  \teval{\sigma, M}{\tau[t]} =&~
  \begin{cases}
    \sigma(\tau)[s_0 .. e_0]\cdot ... \cdot \sigma(\tau)[s_k .. e_k]
                  & \text{if } \teval{\sigma, M}{t} = (s_0, e_0)\cdot ... \cdot(s_k, e_k)\\
                  & \text{and } \forall i. s_i \not = \bot \land e_i \not = \bot\\
    \epsilon      & \text{otherwise}
  \end{cases}\\
\end{align*}

The satisfaction relation $\models$ on \mpe conditions is evaluated w.r.t the
trace assignment $\sigma$ and an m-map~$M$:

\begin{align*}
\sigma, M \models t = t'  \iff &~ \teval{\sigma, M}{t} = \teval{\sigma, M}{t'}\\
\sigma, M \models \neg C  \iff &~ \sigma, M \not \models C\\
\sigma, M \models C \land C' \iff &~ \sigma, M \models C \text{ and } \sigma, M \models C'\\
\end{align*}

\begin{example}
  The condition $l_1 = l_2 \land \tau_1[l_1] \not = \tau_2[l_2]$
gets evaluated to true for $\sigma = \{\tau_1 \mapsto abcba, \tau_2 \mapsto bacab\}$ and
$M = \{l_1,l_2 \mapsto (0, 1)(3,4)\}$, because the terms are evaluated to
$\teval{\sigma, M}{l_1} = \teval{\sigma, M}{l_2} = (0, 1)(3,4)$ and therefore
$\teval{\sigma, M}{\tau_1[l_1]} = abba$ and
$\teval{\sigma, M}{\tau_2[l_2]} = baab$.
\end{example}

Given an \mpe $\alpha$, we denote as $\alpha(\tau)$ the PE associated to
trace variable $\tau$ and if we want to highlight that $\alpha$ has the condition $\varphi$, we write $\alpha[\varphi]$.
We denote the set of all {\mpe}s over trace variables $V_\tau$ and alphabet $\Sigma$ as $\mathit{\mpe}(V_\tau, \Sigma)$.
An \mpe $\alpha[\varphi]$ over trace variables $V_I$ satisfies a trace assignment, written $\sigma \models \alpha$, iff
$\tau \in V_I: \rho(\alpha(\tau), \sigma(\tau)) = (u_\tau, M_\tau, v_\tau)$ and
$\sigma, M_{\tau_1}\mcomp ... \mcomp M_{\tau_k} \models \varphi$ where
    $\{{\tau_1}, ..., {\tau_k}\} = V_I$. That is, $\sigma \models \alpha$ if
every prefix expression for $\alpha$ has matched a prefix and the condition
$\varphi$ is satisfied.




\subsection{Multi-trace prefix transducers}

\emph{Multi-trace prefix transducers (\mpt)} are finite transducers~\cite{VeanesHLMB12} with {\mpe}s as guards on the transitions.
If a transition is taken, one or more symbols are appended to one or more
output words, the state is changed to the target state of the transition and the matched prefixes of input words are consumed. The evaluation then continues matching
new prefixes of the shortened input words.
Combining {\mpe}s with finite state transducers allows to read
input words asynchronously (evaluating {\mpe}s) while having synchronization points
and finite memory (states of the transducer).

\begin{definition}[Multi-trace prefix transducer]
A \emph{multi-trace prefix expression transducer (\mpt)} is a tuple $(V_I, V_O, \Sigma_I, \Sigma_O, Q, q_0, \Delta)$ where
\begin{itemize}
  \item $V_I$ is a finite set of input trace variables
\item $V_O$ is a finite set of output trace variables
\item $\Sigma_I$ is an input alphabet
\item $\Sigma_O$ is an output alphabet
\item $Q$ is a finite non-empty set of states
\item $q_0 \in Q$ is the initial state
\item $\Delta: Q \times \mathit{\mpe}(V_I, \Sigma_I) \times (V_O \hookrightarrow {\Sigma^*_O}) \times Q$
  is the transition relation; we call the partial mappings $(V_O \hookrightarrow {\Sigma^*_O})$ \emph{output assignments}.
\end{itemize}
\label{def:mPE}
\end{definition}

A \emph{run} of an {\mpt} $(V_I, V_O, \Sigma_I, \Sigma_O, L, Q, q_0, \Delta)$
on trace assignment $\sigma_0$ is a sequence 
$\pi = (q_0, \sigma_0) \xrightarrow{\nu_0} (q_1, \sigma_1) \xrightarrow{\nu_1}
\ldots \xrightarrow{\nu_{k-1}} (q_k, \sigma_k)$
of alternating states and trace assignments $(q_i, \sigma_i)$ with output assignments
$\nu_i$,
such that for each $(q_i, \sigma_i) \xrightarrow{\nu_i}(q_{i+1}, \sigma_{i+1})$
there is a transition $(q_i, \alpha, \nu_i, q_{i+1}) \in \Delta$ such that $\sigma_i \models \alpha$ and $\forall\tau \in V_I: \sigma_{i+1}(\tau) = v_\tau$ where
$(\_, \_, v_\tau)=\rho(\alpha(\tau), \sigma_i(\tau))$.
That is, taking a transition in an {\mpt} is conditioned by
satisfying its {\mpe} and its effect is that every matched
prefix is removed from its word and the output assignment is put to output.

The output $O(\pi)$ of the run $\pi$
is the concatenation of the output assignments
$\nu_0 \cdot \nu_1 \cdot ... \cdot \nu_{k-1}$,
where a missing assignment to a trace is considered to be $\epsilon$.
Formally, for any $t\in V_O$, $\nu_i \cdot \nu_j$ takes the value
$$
(\nu_i \cdot \nu_j)(t) =
\begin{cases}
  \nu_i(t)\cdot\nu_j(t) & \text{if } \nu_i(t) \text{ and } \nu_j(t) \text{ are defined}\\
  \nu_i(t)         & \text{if } \nu_i(t) \text{ is defined and } \nu_j(t) \text{ is undefined}\\
  \nu_j(t)         & \text{if } \nu_i(t) \text{ is undefined and } \nu_j(t) \text{ is defined}
\end{cases}
$$

\begin{example} \label{ex:mtpet}


Consider the {\mpt} in Figure~\ref{fig:runs} and
words $t_1 = ababcaba$ and $t_2 = babacbab$ and the assignment
$\sigma = \{\tau_1 \mapsto t_1, \tau_2 \mapsto t_2\}$. The run on this assignment
is depicted in the same figure on the bottom left.
The output of the {\mpt} on $\sigma$ 
is $\bot\top\top\top$.
For any words that are entirely consumed by the {\mpt}
without getting stuck, it holds that $\tau_1$ starts with a sequence
of $ab$'s and $\tau_2$ with a sequence of $ba$'s of the same length.
Then there is one $c$ on both words and the words end with a sequence
of $a$ or $b$ but such that when there is $a$ in one word, there must be $b$
in the other word and vice versa.

\noindent
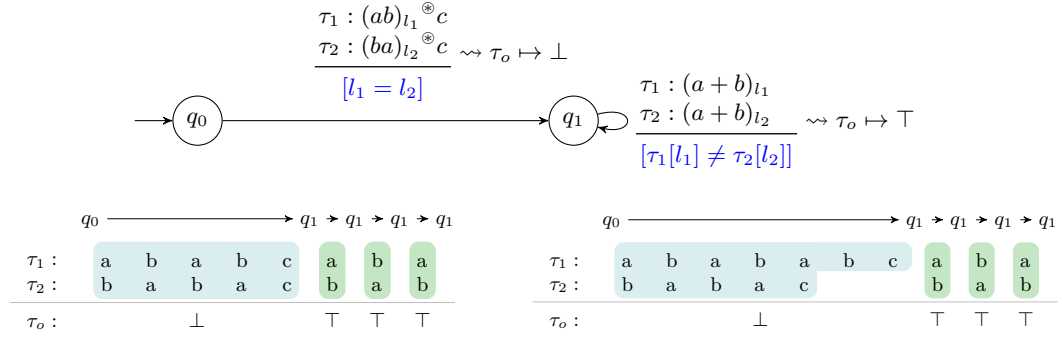
\begin{figure}[t]
\centering
\noindent
\begin{tikzpicture}
  \node[draw,circle, minimum width=6mm] (0) at (0,0) {$q_0$};
  \node[draw,circle, minimum width=6mm] (1) at (5,0) {$q_1$};

  \draw[->] ($(0.west) + (-5mm, 0)$) to (0);
  \draw[->] (0) to node[yshift=9mm,xshift=8mm,align=left]
    {
      \begin{tabular}{l}
      $\tau_1: \iter{\lab{l_1}{ab}}{c}$\\
      $\tau_2: \iter{\lab{l_2}{ba}}{c}$\\
      \midrule
      \multicolumn{1}{c}{\cond{$l_1 = l_2$}}
      \end{tabular}
      $\leadsto \tau_o \mapsto \bot$
    } (1);
    \draw[->] (1) to[loop right] node[yshift=0mm,align=left]
      {
        \small
        \begin{tabular}{l}
        $\tau_1: \lab{l_1}{a + b}$\\
        $\tau_2: \lab{l_2}{a + b}$\\
        \midrule
        \multicolumn{1}{c}{\cond{$\tau_1[l_1] \not = \tau_2[l_2]$}}\\
        \end{tabular}
        $\leadsto \tau_o \mapsto \top$
      } (1);
\end{tikzpicture}
\vspace*{2mm}

\noindent
\begin{minipage}{.41\textwidth}
\begin{center}
\resizebox{\textwidth}{!}
{
\begin{tikzpicture}
  \node[align=left] (t) at  (  -1,0) {\strut $\tau_1:$\\$\tau_2:$};
  \node[align=left] (r1) at (  0,0) {\strut a\\b};
  \node[align=left] (r2) at (0.7,0) {\strut b\\a};
  \node[align=left] (r3) at (1.4,0) {\strut a\\b};
  \node[align=left] (r4) at (2.1,0) {\strut b\\a};
  \node[align=left] (rc) at (2.8,0) {\strut c\\c};
  \node[align=left] (r5) at (3.5,0) {\strut a\\b};
  \node[align=left] (r6) at (4.2,0) {\strut b\\a};
  \node[align=left] (r7) at (4.9,0) {\strut a\\b};

  \draw[color=gray!50] (-1.5, -.5) -- (5.4, -.5);

  \node[align=left] (t0) at (-1, -.8) {\strut$\tau_o:$};
  \node[align=left] (t1) at ( 1.4, -.8) {\strut$\bot$};
  \node[align=left] (t2) at ( 3.5, -.8) {\strut$\top$};
  \node[align=left] (t3) at ( 4.2, -.8) {\strut$\top$};
  \node[align=left] (t4) at ( 4.9, -.8) {\strut$\top$};

  \node[align=left] (q0) at ( -.25, .8) {\strut$q_0$};
  \node[align=left] (q1) at ( 3.15, .8) {\strut$q_1$};
  \node[align=left] (q2) at ( 3.85, .8) {\strut$q_1$};
  \node[align=left] (q3) at ( 4.55, .8) {\strut$q_1$};
  \node[align=left] (q4) at ( 5.25, .8) {\strut$q_1$};

  \draw[->] (q0) to (q1);
  \draw[->] (q1) to (q2);
  \draw[->] (q2) to (q3);
  \draw[->] (q3) to (q4);

  \begin{scope}[on background layer]

  \node[fit=(r1)(rc), inner sep=0, fill=istalightazure, rounded corners]{};

  \node[fit=(r5), inner sep=0, fill=istalightgreen!40, rounded corners]{};
  \node[fit=(r6), inner sep=0, fill=istalightgreen!40, rounded corners]{};
  \node[fit=(r7), inner sep=0, fill=istalightgreen!40, rounded corners]{};
  \end{scope}
\end{tikzpicture}
}
\end{center}
\end{minipage}%
\begin{minipage}{.05\textwidth}
  ~~
\end{minipage}%
\begin{minipage}{.48\textwidth}
\begin{center}
\resizebox{\textwidth}{!}
{
\begin{tikzpicture}
  \node[align=left] (t) at  (  -1,0) {\strut $\tau_1:$\\$\tau_2:$};
  \node[align=left] (r1) at (  0,0) {\strut a\\b};
  \node[align=left] (r2) at (0.7,0) {\strut b\\a};
  \node[align=left] (r3) at (1.4,0) {\strut a\\b};
  \node[align=left] (r4) at (2.1,0) {\strut b\\a};
  \node[align=left] (rc) at (2.8,0) {\strut a\\c};
  \node[align=left] (r5) at (3.5,0) {\strut b\\ };
  \node[align=left] (r6) at (4.2,0) {\strut c\\ };
  \node[align=left] (r7) at (4.9,0) {\strut a\\b};
  \node[align=left] (r8) at (5.6,0) {\strut b\\a};
  \node[align=left] (r9) at (6.3,0) {\strut a\\b};

  \draw[color=gray!50] (-1.5, -.5) -- (6.8, -.5);

  \node[align=left] (q0) at ( -.25, .8) {\strut$q_0$};
  \node[align=left] (q1) at ( 4.55, .8) {\strut$q_1$};
  \node[align=left] (q2) at ( 5.25, .8) {\strut$q_1$};
  \node[align=left] (q3) at ( 5.95, .8) {\strut$q_1$};
  \node[align=left] (q4) at ( 6.65, .8) {\strut$q_1$};

  \node[align=left] (t0) at (-1, -.8) {\strut$\tau_o:$};
  \node[align=left] (t1) at ( 2.1, -.8) {\strut$\bot$};
  \node[align=left] (t2) at ( 4.9, -.8) {\strut$\top$};
  \node[align=left] (t3) at ( 5.6, -.8) {\strut$\top$};
  \node[align=left] (t4) at ( 6.3, -.8) {\strut$\top$};

  \draw[->] (q0) to (q1);
  \draw[->] (q1) to (q2);
  \draw[->] (q2) to (q3);
  \draw[->] (q3) to (q4);

  \begin{scope}[on background layer]

  \node[fit=(r1)(rc), inner sep=0, fill=istalightazure, rounded corners]{};
  \node[inner sep=0, fill=istalightazure, rounded corners,
        minimum height=1.4em, minimum width=2cm] at (3.5, .215) {};

  \node[fit=(r7), inner sep=0, fill=istalightgreen!40, rounded corners]{};
  \node[fit=(r8), inner sep=0, fill=istalightgreen!40, rounded corners]{};
  \node[fit=(r9), inner sep=0, fill=istalightgreen!40, rounded corners]{};

  \end{scope}
\end{tikzpicture}
}
\end{center}
\end{minipage}
\caption{
The \mpt from Example~\ref{ex:mtpet} and a demonstration of its two runs.
Colored regions show parts of words as they are matched by
the transitions, the sequence of passed states is shown above the traces.
}
\label{fig:runs}
\end{figure}

Now assume that the words are
$t_1 = abababcaba$ and $t_2 = babacbab$ with the same assignment.
The situation changes as now the expression on trace $t_1$
matches the prefix $(ab)^3$ while on $t_2$ the prefix $(ba)^2$.
Thus $l_1 = (0, 6) \not = (0, 4) = l_2$ and the match fails.
Finally, assume that we remove the condition $\cond{l_1 = l_2}$
from the first transition. Then for the new words the {\mpt}
matches again and the match is depicted on the bottom right in Figure~\ref{fig:runs}.

\end{example}

In the next section, we work with \emph{deterministic} {\mpt}s.
We say that an \mpt is deterministic if it can take at most one transition
in any situation.
\begin{definition}[Deterministic \mpt]
Let $T = (V_I, V_O, \Sigma_I, \Sigma_O, Q, q_0, \Delta)$ be an {\mpt}.
We say that T is \emph{deterministic (\dmpt)} if for any state $q\in Q$,
and an arbitrary trace assignment $\sigma: V_I \rightarrow \Sigma_I^*$,
if there are multiple transitions $(q, \alpha_1, \nu_1, q_1), ..., (q, \alpha_k, \nu_k, q_k)$ such that $\forall i: \sigma \models \alpha_i$,
it holds that there exists a proper prefix $\eta$ of $\sigma$,
(i.e.,~$\forall \tau\in V_I: \eta(\tau) \le \sigma(\tau)$ and for some $\tau$ $\eta(\tau) <\sigma(\tau)$),
and there exist $i$ such that
$\eta \models \alpha_i$ and $\forall j \not = i: \eta \not \models \alpha_j$.
\end{definition}

Intuitively, an \mpt is \dmpt if whenever there is a trace assignment that
satisfies more than one transition from a state, one of the transitions matches
``earlier'' than any other of those transitions.








\section{Hypertrace transformations} \label{sec:hyper}

A \emph{hyperproperty} is a set of sets of infinite traces.
In this section, we discuss an algorithm for monitoring \emph{\ksafety hyperproperties},
which are those whose violation can be witnessed by at most $k$ finite traces:
\begin{definition}[\ksafety hyperproperty]A hyperproperty $S$ is \ksafety hyperproperty 
  iff
\[
  \forall T \subseteq  \Sigma^\omega: T \not \in S \implies
  \exists M \subseteq \Sigma^*: M \le T \land |M| \le k \land
  (\forall T' \subseteq\Sigma^\omega: M \le T' \implies T' \not \in S)
\]
\end{definition}
where $\Sigma^\omega$ is the set of infinite words over alphabet $\Sigma$,
and $M \le T$ means that each word in $M$ is a (finite) prefix of a word in $T$.

We assume \emph{unbounded parallel input model}~\cite{FinkbeinerHST19},
where there may be arbitrary many traces digested in parallel, and
new traces may be announced at any time.
Our algorithm is basically the combinatorial
algorithm for monitoring hyperproperties of Finkbeiner~et~al.~\cite{FinkbeinerHST17,Hahn19} where we exchange automata generated from HyperLTL
specifications with {\mpt}s.
That is, to monitor a \ksafety hyperproperty, we instantiate an \mpt
for every k-tuple of input traces.
%
An advantage of using {\mpt}s instead of monitor automata in the algorithm of
Finkbeiner~et~al. is that we automatically
get a monitoring solution for asynchronous hyperproperties.
A disadvantage is that we cannot automatically use some of the optimizations designed
for HyperLTL monitors that are based on the structure (e.g., symmetry) of the HyperLTL
formula~\cite{FinkbeinerHST18}.

The presented version of our algorithm assumes that the input is \dmpt
as it is preferable to have deterministic monitors.
Deciding whether a given \mpt is deterministic depends a lot on the chosen logic used
for \mpe constraints. We are currently working on an algorithm that decides
the determinism for {\mpt}s with the equality constraints used in this paper.
In the rest of the paper, we assume that the used {\mpt}s are \emph{known} to be
{\dmpt}s (which is also the case of {\mpt}s used in our evaluation).
We make the remark that in cases where the input \mpt is not known to be
deterministic and/or a check is impractical,
one may resort to a way how to resolve possible non-determinism instead,
such as using priorities on the edges.
This is a completely valid solution and it is easy to modify our algorithm to work
this way. In fact, the algorithm works also with non-deterministic {\mpt}s with a small modification.

\subsection{Algorithm for online monitoring of \ksafety hyperproperties}

Our algorithm is depicted in Algorithm~\ref{alg:combinatorial}
and Algorithm~\ref{alg:aux} (auxiliary procedures).
In essence, the algorithm maintains a set of \emph{configurations} where one
configuration corresponds to the state of evaluation of one edge of an {\dmpt} instance.
Whenever the algorithm may make a progress in some configuration,
it does so and acts according to whether matching the edge succeeds, fails, or needs more events.

\begin{algorithm}
  \DontPrintSemicolon
  \KwData{an \dmpt $(\{\{\tau_1, ...,\tau_k\}, \{\tau_O\}, \Sigma_I, \{\bot, \top\}, Q, q_0, \Delta\})$}
  \KwResult{\emph{false} + witness if an \dmpt instance outputs $\bot$, \emph{true} if no \dmpt instance outputs $\bot$ and there is finitely many traces. The algorithm does not terminate otherwise.}
  \vspace*{.8em}

  $\val{traces} \gets \emptyset$ \tcp*{Stored contents of all traces}
  $\val{onlinetraces} \gets \emptyset$ \tcp*{Traces that are still being extended}
  $\val{workbag} \gets \emptyset$ \tcp*{Sets of configurations to process}
  \;

  \SetKwFunction{Fcfgs}{cfgs}
  \SetKwFunction{Fupdate}{update\_traces}
  \SetKwProg{Proc}{Procedure}{}{}

  \While {$true$}  { \label{alg:loop}
    \Fupdate($\val{workbag}$, $\val{onlinetraces}$, $\val{traces}$)\label{alg:call_update}
    \;
    $\val{workbag'} \gets \emptyset$ \tcp{The new contents of workbag}
    \;
    \ForEach {$C \in \val{workbag}$} { \label{alg:iter_workbag}
      $\val{workbag} \gets \val{workbag} \setminus \{C\}$\label{alg:pop_workbag}\;
      $C' \gets \emptyset$ \tcp{Rewritten set of configurations}\;

      \tcp{Try to move each configuration in the set of configurations}
      \ForEach {c = ($\sigma, (p_1, ..., p_k), M, q \xrightarrow{E\cond{\varphi}\leadsto\nu} q')) \in C$} { \label{alg:foreach_cfg}
      \tcp{Progress on each trace where possible}
      $E', M' \gets E, M$\;
      $(p'_1, ..., p'_k) \gets (p_1, ..., p_k)$\;
      \ForEach {$1 \le i \le k$ s.t. $p_i < |\val{traces}(\sigma(\tau_i))|
           \land E(\tau_i) \not = \epsilon$} { \label{alg:foreach_progress}
        $E' \gets E'[\tau_i \mapsto \xi]$ where
          $E(\tau_i), M' \step{\val{traces}(\sigma(\tau_i))[p_i]}{p_i} \xi, M''$
        \label{alg:step} \;
          $M' \gets M''$\;
          $p'_i \gets p'_i + 1$\;
        \If(\tcp*[f]{Configuration failed}) {$\xi = \bot$}{ \label{alg:failed}
          \KwCont with next configuration (line \ref{alg:foreach_cfg}) \label{alg:continue_cfg}\;
        }
    }
    \uIf (\tcp*[f]{All prefix expressions matched})
      {$\forall j. E'(j) = \epsilon$} { \label{alg:mpe_match}
      \eIf(\tcp*[f]{The condition is satisfied})
      {$\sigma, M' \models \varphi$} { \label{alg:cond}
        \tcp{Compare $p'_1, ..., p'_k$ against other configurations from this set to see if this must be the shortest match}
        \If {$(p'_1, ..., p'_k) < (p''_1, ..., p''_k)$ for any $(p''_1, ..., p''_k)$ of others $c' \in C$} { \label{alg:shortest}
          \If(\tcp*[f]{Violation found}) {$\bot \in \nu$}{ \label{alg:bot}
              \KwRet{\emph{false} + $\sigma$} \label{alg:false}
          }
          \tcp{Edge is matched, no violation found, queue successor edges}
          $\val{workbag'} \gets \val{workbag'} \cup
          \{\Fcfgs(q',
                   (\sigma(\tau_1), ..., \sigma(\tau_k)) ,
                   (p'_1, ..., p'_k)
                 )\}$ \label{alg:queue_succ} \;
          \tcp{This set of configurations is done}
          \KwCont outer-most loop (line \ref{alg:iter_workbag}) \label{alg:continue}\;

        }
      }{
        \KwCont with next configuration (line \ref{alg:foreach_cfg}) \label{alg:continue_cfg2}\;
      } 
    }
    \tcp{Check if the configuration has matched or it can still make a progress,
         and in that case put it back (modified) to the set}
    \If{$E'$ has matched or $\neg\left(\forall 1 \le i \le k: \sigma(\tau_i) \not\in\val{onlinetraces}
      \land p_i = |\val{traces}(\sigma(\tau_i))|\right)$}{\label{alg:cfg_done}
        $C' \gets C' \cup \{(\sigma,
                          (p'_1, ..., p'_k),
                          M',
                           q \xrightarrow{E'\cond{\varphi}\leadsto\nu} q'
                         )\}$ \label{alg:modify_cfg}\;

    }
  }
  \If {$C' \not = \emptyset$ } {
      \tcp{Queue the modified set of configurations}
      $\val{workbag'} \gets \val{workbag'} \cup \{C'\}$
      \label{alg:requeue_workbag}
      \;
      }
  }

    \;
    $\val{workbag} \gets \val{workbag'}$\;
    \If {$\val{workbag} = \emptyset$ and no new trace will appear} {\label{alg:term}
      \KwRet{\emph{true}}\label{alg:true}
    }
  }
\caption{Online algorithm for monitoring hyperproperties with {\mpt}s}
  \label{alg:combinatorial}
\end{algorithm}

\begin{algorithm}
  \DontPrintSemicolon
  \SetKwFunction{Fcfgs}{cfgs}
  \SetKwFunction{Fupdate}{update\_traces}
  \SetKwProg{Proc}{Procedure}{}{}

  \tcp{Auxiliary procedure that returns a set of configurations for outgoing edges of $q$}
  \Proc{\Fcfgs{$q$, $(t_1, ..., t_k)$, $(p_1, ..., p_k)$}}{
    $\sigma \gets \{\tau_i \mapsto t_i \mid 1 \le i \le k \}$\;

    \KwRet \{$(\sigma, (p_1, ..., p_k), (0, ..., 0), \emptyset, e) \mid e $ is an outgoing edge from $q$ \}\;
  }
  \;
  \tcp{Auxiliary procedure to add new traces and update the current ones}
  \Proc{\Fupdate{$\val{workbag}$,
                 $\val{onlinetraces}$,
                 $\val{traces}$}}
  {\label{alg:aux_update}
    \If(\tcp*[f]{Update $\val{traces}$ and $\val{workbag}$ with the new trace})
    {there is a new trace $t$}
    { \label{alg:new_trace_if}
      $\val{onlinetraces} \gets \val{onlinetraces} \cup \{t\}$\;
      $\val{traces} \gets \val{traces}[t \mapsto \epsilon]$\;
      $\val{tuples} \gets
         \{(t_1, ..., t_k) \mid t_j \in Dom(traces), t = t_i \text{ for some } i\}$

      $\val{workbag} \gets
         \val{workbag} \cup
         \{\Fcfgs(q_0, (t_1, ..., t_k), (0, ..., 0))\mid (t_1, ..., t_k) \in tuples
         \}$ \label{alg:init_cfgs}
    }
    \;
    \ForEach(\tcp*[f]{Update traces with new events})
    {$t \in \val{onlinetraces}$ that has a new event $e$}
    {
      $\val{traces}(t) =  \val{traces}(t)\cdot e$\label{alg:trace_update}\;
      \If(\tcp*[f]{Remove finished traces from $\val{onlinetraces}$})
      {$e$ was the last event on $t$} {
        $\val{onlinetraces} \gets \val{onlinetraces} \setminus \{t\}$
        \label{alg:onlinetraces_remove}\;
      }
    }
  }

  \caption{Auxiliary procedures for Algorithm~\ref{alg:combinatorial}}
  \label{alg:aux}
\end{algorithm}

Now we discuss functioning of the algorithm in more detail.
The input is an {\dmpt} $(\{V_I, \{\tau_O\}, \Sigma_I,$ $ \{\bot, \top\}, Q, q_0, \Delta\})$. W.l.o.g we assume that $V_I = \{\tau_1, ..., \tau_k\}$.
The {\dmpt} outputs a sequence of verdicts $\{\top, \bot\}^*$.
A violation of the property is represented by $\bot$, so
whenever $\bot$ appears on the output, the algorithm terminates and reports the violation.
\todo{A simple modification is to terminate on true.}

A configuration is a 4-tuple $(\sigma, (p_1, ..., p_k), M, e)$ where
$\sigma$ is a function that maps trace variables to traces,
the vector $(p_1, ..., p_k)$ keeps track of
reading positions in the input traces, $M$ is the current m-map gathered
while evaluating the \mpe of $e$, and $e$ is the edge that is being evaluated.
More precisely, $e$ is the edge that still needs to be evaluated in the future
as its \mpe gets repeatedly rewritten during the run of the algorithm.
If the edge has \mpe $E$, we write
$E[\tau \mapsto \xi]$ for the \mpe created from $E$ by setting the \pe
for $\tau$ to $\xi$.

The algorithm uses three global variables.
Variable $\val{workbag}$ stores configurations to be processed.
Variable $\val{traces}$ is a map that remembers the so-far-seen contents of all traces.
Whenever a new event arrives on a trace $t$,
we append it to $\val{traces}(t)$. 
Traces on which a new event may still arrive are stored in variable
$\val{onlinetraces}$.
Note that to follow the spirit of online monitoring setup,
in this section, we treat traces as \emph{opaque} objects
that we query for next events.

In each iteration of the main loop (line~\ref{alg:loop}), the algorithm
first calls the procedure \texttt{update\_traces} (line~\ref{alg:call_update},
the procedure is defined in Algorithm~\ref{alg:aux}).
This procedure adds new traces to $\val{onlinetraces}$ and
updates $\val{workbag}$ with new configurations if there are any new traces,
and extends traces in $\val{traces}$ with new events.
The core of the algorithm are lines~\ref{alg:iter_workbag}--\ref{alg:requeue_workbag}
that take all configuration sets and update them with unprocessed events.

The algorithm pops every set of configurations from $\val{workbag}$
(lines~\ref{alg:iter_workbag}--\ref{alg:pop_workbag})
and attempts to make a progress on every configuration in the set (line~\ref{alg:foreach_cfg}).
For each trace where a progress can be made in the current configuration
(line~\ref{alg:foreach_progress}),
i.e., there is an unprocessed
event on the trace $\tau_i$ ($p_i < |\val{traces}(\sigma(\tau_i))|$), and the
corresponding \pe on the edge still has not matched ($E(\tau_i) \not =
\epsilon$), we do a step on this \pe (line~\ref{alg:step}).
The new state of the configuration is aggregated into the primed temporary variables
($E', M', ...$).
If the {\mpe} matches (lines~\ref{alg:mpe_match} and \ref{alg:cond}),
we check if other configurations from the set have progressed enough for us to be
sure that this configuration has matched the shortest prefix (line~\ref{alg:shortest}).
That is, we compare $p'_1, ..., p'_k$ against positions $p''_1, ..., p''_k$ from each other configuration in $C$ if it is strictly smaller
(i.e.,~$p'_i \le p''_i$ for all $i$ and there is $j$ s.t., $p'_j < p''_j$).
If this is true, we can be sure that there is no other edge that
can match a shorter prefix and that has not matched it yet because it was waiting for events.
If this configuration is the shortest match,
the output of the edge is checked if it contains $\bot$ (line~\ref{alg:bot})
and if so, \emph{false} with the counterexample is returned on line~\ref{alg:false}
because the monitored property is violated.
Else, the code falls-through to line~\ref{alg:queue_succ} that queues
new configurations for successor edges as the current edge has been successfully matched
and then continues with a new iteration of the outer-most loop (line~\ref{alg:continue}).
The continue statement has the effect that all other configurations derived from the
same state (other edges) are dropped and therefore progress is made only on the
configuration (edge) that matched.
If any progress on the \mpe can be made in the future, or it has already matched
but we do not know if it is the shortest match yet,
the modified configuration is pushed into the set
of configurations instead of the original one (line~\ref{alg:modify_cfg}).
If not all the configurations from $C$ were dropped because they could not proceed,
line~\ref{alg:requeue_workbag} pushes the modified set of configurations back to
workbag and new iteration starts.

\subsection{Discussion}


To see that the algorithm is correct, let us follow the evolution of the set
of configurations for a single instance of the {\dmpt} on traces $t_1, ..., t_k$.
The initial set of configurations corresponding to outgoing edges from
the initial state is created and put to $\val{workbag}$
exactly once on line~\ref{alg:init_cfgs} in Algorithm~\ref{alg:aux}.
When it is later popped from \emph{workbag} on lines~\ref{alg:iter_workbag}--\ref{alg:pop_workbag} (we are back in Algorithm~\ref{alg:combinatorial}), every configuration (edge) is updated -- a step is taken on every
\pe from the edge's \mpe (lines~\ref{alg:foreach_progress}--\ref{alg:step})
where a step can be made.
If matching the \mpe fails, the configuration is discarded due to the jump on line~\ref{alg:continue_cfg}
or line~\ref{alg:continue_cfg2}.
If matching the \mpe has neither failed nor succeeded (and no violation has been found,
in which case the algorithm would immediately terminate),
the updated configuration is pushed back to $\val{workbag}$ and revisited in later iterations.
If the \mpe has been successfully matched and it is not known to be the shortest match, it is put back
to $\val{workbag}$ and revisited later when other configurations may have proceeded and we
may again check if it is the shortest match or not.
If it is the shortest match, its successor edges are queued to $\val{workbag}$
on line~\ref{alg:queue_succ} (if no violation is found).
This way we incrementally first match the first edge on the run of
the \dmpt on traces $t_1, ..., t_k$ (or find out that no edge matches),
then the second edge after it gets queued into $\val{workbag}$
on line~\ref{alg:queue_succ}, and so on.

The algorithm terminates if the number of traces is bounded.
If it has not terminated because of finding a violation on line~\ref{alg:false},
it will terminate on line~\ref{alg:true}.
To see that the condition on line~\ref{alg:term} will eventually get true if the number
of traces is bounded, it is enough to realize that unless a configuration
gets matched or failed, it is discarded at latest when failing the condition
on line~\ref{alg:cfg_done} after reading entirely (finite) input traces.
Otherwise, if a configuration fails, the set is never put back to $\val{workbag}$
and if it gets matched, it can get back to $\val{workbag}$ repeatedly
only until the shortest match is identified. But if
every event comes in finite time, some of the configurations in the set will eventually
be identified as the shortest match (because the \mpt is deterministic), and the set of configurations will be done.
Therefore, $\val{workbag}$ will eventually become empty.

Worth remark is that if we give up on checking if the matched \mpe is the shortest
match on line~\ref{alg:shortest} (we set the condition to \emph{true})
and on line~\ref{alg:continue}, we continue with the loop
on line~\ref{alg:foreach_cfg} instead of with the outer-most loop, i.e.,
we do not discard the set of configurations upon a successfully taken edge,
the algorithm will work also for generic \emph{non-deterministic} {\mpt}s.

There is another difference between our algorithm and the algorithm of
Finkbeiner~et~al.~\cite{FinkbeinerHST17,FinkbeinerHST19,Hahn19}.
In our algorithm, we assume that existing traces may be extended at any time
until the last event has been seen. 
The algorithm of Finkbeiner~et~al. assumes that when a new trace appears, its
contents is entirely known.
So their algorithm is incremental on the level of traces, while our algorithm
is incremental on the level of traces \emph{and} events.

The monitor templates in the algorithm of Finkbeiner~et~al.
are automata whose edges match propositions on different traces. Therefore, they can
be seen as trivial {\dmpt}s where each prefix expression is a single letter or $\epsilon$.
Realizing this, we could say that our monitoring algorithm is an asynchronous
extension of the algorithm of Finkbeiner~et~al. where we allow to read multiple
letters on edges, or, alternatively, that in the context of monitoring
hyperproperties, {\dmpt}s are a generalization of HyperLTL template automata
to asynchronous settings.

\section{Empirical evaluation}  \label{sec:experiments}

We conducted a set of experiments about monitoring
asynchronous version of OD on random and semi-random traces.
The traces contain input and output events
\texttt{I(t, n)} and \texttt{O(t, n)}
with \texttt{t} $\in \{$\texttt{l}$, $ \texttt{h}$\}$,
and \texttt{n} a 64-bit unsigned number.
Further, a trace can contain
the event \texttt{E} without parameters that abstracts any event
that have occurred in the system, but that is irrelevant to OD.

\begin{figure}[t]
\begin{center}
  \begin{tikzpicture}[
    q/.style = {scale=1, draw, circle, minimum width=7mm, align=left},
    lab-edge/.style = {scale=0.8, align=left}
    ]
  \node[q] (0) at (0,0) {$q_{0}$};
  \node[q] (1) at (5,0) {$q_{1}$};
  \node[q] (2) at (-8.5,0) {$q_{2}$};

  \draw[->] ($(0.south) + (0, -3mm)$) to (0);
  \draw[->] (0) to[loop] node[lab-edge, yshift=10mm, xshift=0mm]
    {
      \begin{tabular}{l}
        $\tau_1: \iter{E}{\lab{e_1}{I_l + O_l}}$\\
        $\tau_2: \iter{E}{\lab{e_2}{I_l + O_l}}$\\
      \midrule
      \multicolumn{1}{c}{\cond{$\tau_1[e_1] = \tau_2[e_2]$}}
      \end{tabular}
    }
    (0);
  \draw[->] (0) to node[lab-edge, yshift=10mm, xshift=4mm]
    {
      \begin{tabular}{l}
        $\tau_1: \iter{E}{\lab{e_1}{O_l + \$}}$\\
        $\tau_2: \iter{E}{\lab{e_2}{O_l + \$}}$\\
      \midrule
      \multicolumn{1}{c}{\cond{$\tau_1[e_1] \not = \tau_2[e_2]$}}
      \end{tabular}
      $\leadsto \tau_o \mapsto \bot$
    }
    (1);
  \draw[->] (0) to node[lab-edge, yshift=16mm, xshift=0mm]
    {
      \begin{tabular}{l}
        $\tau_1: \iter{E}{\lab{e_1}{O_l + I_l + \$}}$\\
        $\tau_2: \iter{E}{\lab{e_2}{O_l + I_l + \$}}$\\
      \midrule
      \multicolumn{1}{c}{\cond{
          \begin{tabular}{l}
          $\tau_1[e_1] \not = \tau_2[e_2]~\land$\\
          $\tau_1[e_1] = O_l \implies \tau_2[e_2] \not \in \{O_l, \$\}~\land$\\
          $\tau_2[e_2] = O_l \implies \tau_1[e_1] \not \in \{O_l, \$\}$
          \end{tabular}
          }
        }
      \end{tabular}
      $\leadsto \tau_o \mapsto \top$
    }
    (2);

\end{tikzpicture}
\end{center}
\caption{The \dmpt used for monitoring asynchronous OD in the experiments.}
\label{fig:mpt_ex}
\end{figure}
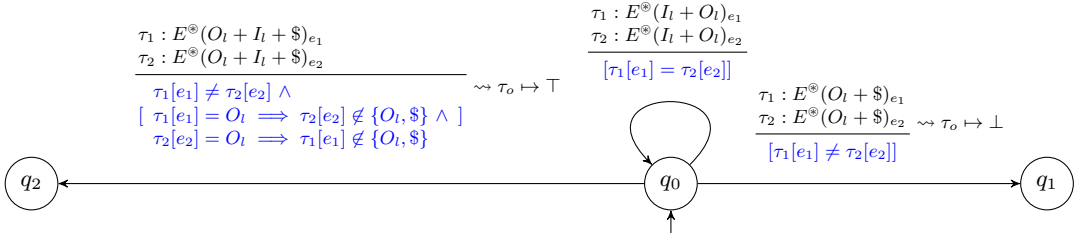

The \dmpt used for monitoring OD is
a modified version of the \dmpt for monitoring OD from Section~\ref{sec:intro}, and is shown in Figure~\ref{fig:mpt_ex}.
The modification makes the \dmpt handle also traces with different number and order of input and output events.
The letter $\$$ represents the end of trace and is automatically appended to the input traces.
We abuse the notation and write $\tau_1[e_1] = O_l$ for the expression that would be formally a disjunction comparing $\tau_1[e_1]$ to all possible constants represented by $O_l$. However, in the implementation, this is a simple constant-time check of the type of the event, identical to checking that an event matches $O_l$ when evaluating prefix expressions.
The term $\tau_i[e_i] \not \in \{O_l, \$\}$ is just a shortcut for 
$\tau_i[e_i] \not = O_l \land \tau_i[e_i] \not = \$$.

The self-loop transition in the \dmpt in Figure~\ref{fig:mpt_ex} has no output and we enabled the algorithm to stop processing traces whenever $\top$ is detected on the output of the transducer
because that means that OD holds for the input traces.
Also, we used the reduction of traces~\cite{FinkbeinerHST18} -- because OD is symmetric and reflexive, then if we evaluate it on the tuple of traces $(t_1, t_2)$, we do not have to evaluate it for $(t_2, t_1)$ and $(t_i, t_i)$.

\begin{figure}[t]
  \centering
  \begin{tabular}{c}
    \includegraphics[width=\textwidth]{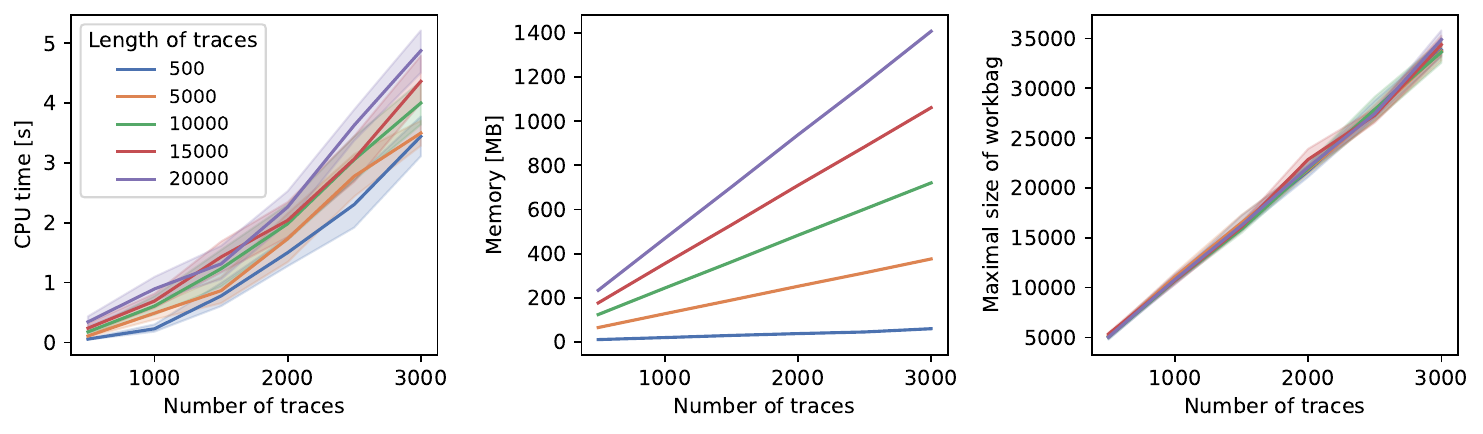}
  \end{tabular}
  \caption{CPU time and maximal memory consumption of monitoring asynchronous OD on random traces with approx. 10\% of low input events and 10\% of low output events. Values are the average of 10 runs.}
  \label{fig:plots1}
\end{figure}

Monitors were implemented in C++ with the help of the framework
\vamos~\cite{ChalupaMLH23}. The experiments were run on a machine
with \emph{AMD EPYC} CPU with the frequency 3.1\,GHz.
An artifact with the implementation of the algorithm and scripts to reproduce the experiments can be found on Zenodo\footnote{
\doi{10.5281/zenodo.8191723}
}.

\subsubsection{Experiments on random traces}
In this experiment, input traces of different lengths were generated such that approx. 10\% were low input and 10\% low output events.
These events always carried the value 0 or 1 to increase the chance
that some traces coincide on inputs and outputs.

Results of this experiment are depicted in Figure~\ref{fig:plots1}.
The left plot shows that the monitor is capable of processing hundreds of traces in a short time and seem to scale well with the number of traces, irrespective of the length of traces.
The memory consumption is depending more on the length of traces as shown in
the middle plot. This is expected as all the input traces are stored in memory.
Finally, the maximal size of the workbag grows linearly with the number of traces
but not with the length of traces, as the right plot shows.

\subsubsection{Experiments on periodic traces}
In this experiment, we generated a single trace that contains low input and output events periodically with fixed distances.
Multiple instances of this trace were used as the input traces.
The goal of the experiment is to see how the monitor performs on traces that must be processed to the very end and if the performance is affected by the layout of events.

The plots in Figure~\ref{fig:plots2} show that
the monitor scales worse than on random traces as it has to always process the traces to their end. For the same reason, the performance
of the monitor depends more on the length of the traces.
Still, it can process hundreds of traces in a reasonable time.
The data do not provide a clear hint on how the distances between events change the runtime, but they do not affect it significantly.
The memory consumption remains unaffected by distances.







\begin{figure}[t]
  \centering
  \begin{tabular}{c}
    \includegraphics[width=\textwidth]{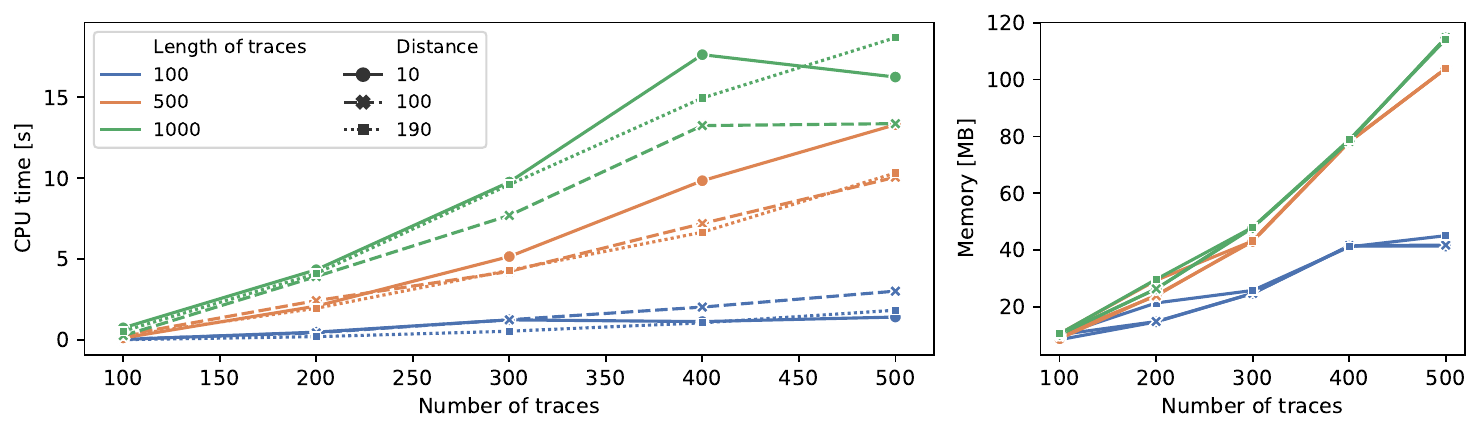}
  \end{tabular}
  \caption{The plot shows CPU time and memory consumption of monitoring asynchronous OD on instances of the same trace with low input and output events laid out periodically with fixed distances. }
  \label{fig:plots2}
\end{figure}

\section{Related Work}\label{sec:related}
In this section, we review the most closely related work.
More exhaustive review can be found in the extended version of this paper~\cite{extended}.

\bigskip
\noindent
\textbf{Logics for hyperproperties}~
Logics for hyperproperties are typically created by extending a logic for trace
properties. Hyperlogics \emph{HyperLTL}~\cite{ClarksonFKMRS14},
\emph{HyperCTL$^*$}~\cite{ClarksonFKMRS14}, and
\emph{HyperQPTL}~\cite{Raabe16} extend \emph{LTL}~\cite{Pnueli77},
\emph{CTL$^*$}~\cite{EmersonH86}, and \emph{QPTL}~\cite{Sistla83}, resp.,
with explicit quantification over traces.
The logic \emph{FO[<,E]}~\cite{Finkbeiner017} and \emph{S1S[E]}~\cite{CoenenFHH19}
are first- and second- order logics with successors
extended with the \emph{equal level predicate}
that relates the same time points on different traces.
\emph{TeamLTL}~\cite{KrebsMV018,VirtemaHFK021} extends LTL with \emph{team semantics},
where teams are sets of traces.
Many other formalism were introduced to describe hyperproperties~\cite{DimitrovaFKRS12,BozzelliMP15,CoenenFHH19,GutsfeldMO20,BonakdarpourS21,BartocciFHNC22,AcetoAAF22}.

All of the hitherto mentioned logics use \emph{synchronous time model}.
\emph{Asynchronous HyperLTL}~\cite{BaumeisterCBFS21},
\emph{Stuttering HyperLTL}~\cite{BozzelliPS21}, and \emph{Context HyperLTL}~\cite{BozzelliPS21}
are extensions of HyperLTL to asynchronous time model.
Gutsfeld~et~al.~\cite{GutsfeldMO21} introduce very expressive
\emph{Multi-tape Alternating Asynchronous Word Automata (AAWA)}
and the temporal fixpoint calculus $H_\mu$ for the specification and analysis
of asynchronous hyperproperties. AAWAs are
so far the only automata-based formalism for specification of asynchronous hyperproperties.
Beutner~et~al. define \emph{HyperATL$^*$}~\cite{BeutnerF21},
an extension of the logic \emph{ATL$^*$}~\cite{AlurHK02} that can capture
asynchronous hyperproperties via quantification over \emph{strategies} of a scheduler.
\emph{Hypernode automata}~\cite{bartocci2023hypernode} introduced by Bartocci~et~al. 
combine finite-state automata with the \emph{hypertrace logic} which allows to describe
properties that have multiple ,,phases''. The hypertrace logic ignores stuttering and prefixing to enable asynchronous time model.

\medskip
\noindent
\textbf{Runtime monitoring of hyperproperties}~
The first paper on runtime monitoring of hyperproperties
is due to Agrawal and Bonakdarpur~\cite{AgrawalB16}.
They consider monitoring \ksafety hyperproperties specified with HyperLTL.
In general, monitoring algorithms for hyperproperties can be classified
as combinatorial or constraint-based~\cite{Hahn19}.
Combinatorial algorithms~\cite{FinkbeinerHST17,FinkbeinerHST18,Hahn19}
construct multiple instances of an monitoring automaton and therefore our
algorithm fall into this category.
Constraint-based algorithms~\cite{BrettSB17,HahnST19,Hahn19,AcetoAAF22}
translate the monitoring task into a set of constraints (e.g., SMT formulae)
and apply rewriting and solving  of the constraints to monitor a given hyperproperty.
%
%

\emph{Stream runtime verification (SRV)}~\cite{Sanchez21} specifies monitoring
as transformation of streams of data, which makes SRV also related to transducers.
It is common that there are multiple input and output streams in an SRV specification,
and
languages like TeSSLa~\cite{LeuckerSS0S18} support
asynchronous time model. So far as we know, no one has used
SRV languages in the context of hyperproperties yet.

\emph{BeepBeep}~\cite{HalleK17} and \vamos~\cite{ChalupaMLH23}
are stream processors designed not only for runtime monitoring.
Both frameworks support an unbounded number of input streams.
\vamos allows their dynamic creation, and as for BeepBeep,
one can use its API to program the same functionality.
Neither of the frameworks was considered in the context of hyperproperties yet
-- we used some data structures from \vamos in the implementation of our monitors,
but we have not used its capabilities as a stream processor.

There are many works on the monitorability and the complexity of monitoring
hyperproperties~\cite{AgrawalB16,BonakdarpourF18,BonakdarpourSS18,FinkbeinerHST19,StuckiSSB19,DamanafshanF22}.

\sidenote{
- combination with static analysis~\cite{BonakdarpourSS18} to tackle the monitorability problem
- reduction to monitoring trace properties~\cite{PinisettySS18}, focused on 2-safety, which includes "1 variants of many interesting security properties including noninteference \cite{ClarksonFKMRS14}, integrity
[10], software doping [5, 13], and data minimality [3]."
- alternation-free fragment of HyperLTL\cite{BrettSB17}, space independent of the number of traces
}

\medskip
\noindent
\textbf{Runtime enforcement}~
\emph{Runtime enforcement}~\cite{Schneider00,FalconeMFR11} is the process of
transforming input streams into output streams such that the output streams
satisfy a given property, usually a security policy.
Runtime enforcement of hyperproperties was explored by Coenen~et~al.~\cite{CoenenFHHS21}.
Prefix transducers also transform input streams to output streams and thus
could be used to enforce (at least some) policies\footnote{
At least on the abstract level, the enforcement has usually also the other part that actually modifies the running system.
}.











\medskip
\noindent
\textbf{Automata and regular expressions}~
\emph{Automata}~\cite{HopcroftU79} and \emph{transducers}~\cite{VeanesHLMB12} are the basis of {\mpt}s and are well explored.
PEs are closely tied with the co-inductive definition of automata~\cite{Rutten98} because
both are based on language derivatives~\cite{Brzozowski64,Antimirov96}.
In fact, the decomposition function $\rho$ can be directly used 
to define the co-algebra of {\pe}s.
\emph{Multi-track automata}~\cite{Bultan2017} are automata that read n-tuples of letters.
They commonly use also a special letter $\lambda$ for a gap (no letter) and thus can describe asynchronous reading of words.
\emph{Regular expressions (RE)}~\cite{HopcroftU79} are an ubiquitous
formalism with many uses and many restrictions/extensions.
\emph{Prefixed regular expressions (PRE)}~\cite{Baeza-YatesG96}
are a subset of REs
with some properties similar to PEs.
Semantically, \emph{prefix-free} REs~\cite{HanWW05} are closer to PEs than PREs,
because PEs give raise to prefix-free languages as follows from~\cite[Lemma 1]{HanW04} 
and the fact that a PE corresponds to a prefix-free transducer~\cite{extended}.
REs with the \emph{shortest-match semantics}~\cite{ClarkeC97}
are very close to PEs, however, unlike PEs, they can be ambiguous.
{\mpt}s with a single input word could be seen as a modification of
\emph{expression automata}~\cite{HanW04}, which are automata with REs on edges.

{\pe}s use binary version of Kleene star.
The ``original'' star defined by Kleene was actually binary~\cite{Kleene,BergstraBP94}.
Kleene, however, did not have restrictions on the right-hand side of the star
as we do and it developed in a short time into the more engaging unary star that people use today~\cite{Copi58}.
\emph{Backreferences} in regular expressions refer to
parts of the word that were already matched~\cite{PennaITZ03,ChidaT22}.
They can be even named~\cite{BerglundM17}, raising more similarities to our labels.
In REs, backreferences bring a great power as they allow
non-regular languages to be matched~\cite{BerglundM17}.
{\pe}s can also recognize some non-regular patterns, however,
labels are much weaker than backreferences because \mpe constraints
are evaluated only a posteriori, while backreferences modify the way how REs are matched.

\todo{unambiguity in formalism}

{
  \footnotesize








}

\section{Conclusion and Future Work}\label{sec:conclusion}

We introduced prefix expressions and multi-trace prefix transducers, a formalism that we see as a natural executable specification
for the monitoring of synchronous and asynchronous hyperproperties.
Prefix expressions are similar to regular expressions, but match only prefixes of words.
The reason why we prefer prefix expressions over regular expressions (that could also be used to match prefixes) is that our
prefix expressions are deterministic and unambiguous.
These properties make evaluating prefix expressions efficient.
The matched prefixes, more precisely their parts that were explicitely labeled, can be then reasoned about using logical formulae, which are a part of multi-trace
prefix expressions that extend prefix expressions to multiple words.
Multi-trace prefix expressions are used as guards on edges in multi-trace prefix transducers, which incrementally match and consume
prefixes of input words and transform them into output words.
Combining prefix expressions with finite state transducers allows us to read input words asynchronously (matching prefix expressions)
with synchronisation points (states of the transducer).
Combining finite state transducers with a ''labeled`` version of some other formalism like LTL or regular expressions would work similarly well if one solves or accepts possible non-determinism and unambiguity of these formalisms.

We use prefix transducers to monitor synchronous and asynchronous \ksafety hyperproperties.
Our experimental evaluation of monitoring asynchronous observational determinism shows that a prefix-transducer-based monitoring
algorithm can scale to thousands of traces.

Prefix transducers provide a flexible formalism for optimizing monitoring algorithms.
We currently implement an asynchronous monitoring algorithm that uses prefix transducers to summarize the seen traces,
similar to the constraint-based algorithms for monitoring synchronous hyperproperties~\cite{Hahn19}.
We also want to analyze the transducers to avoid instantiating them on redundant tuples of traces,
similar to the optimizations for HyperLTL monitors~\cite{FinkbeinerHST18}.
Furthermore, the evaluation of prefix transducers provides many opportunities for parallelization,
ranging from parallelizing the workbag in Algorithm~\ref{alg:combinatorial} to evaluating prefix expressions for different traces in parallel.
Finally, we work on compiling prefix transducers from a high-level logical specification languages for asynchronous hyperproperties,
namely~\cite{bartocci2023hypernode}.
All our implementation work is carried out to extend the \vamos~\cite{ChalupaMLH23} software infrastructure for monitoring.

\sidenote{
We could adapt known formalisms, for example LTL or regular expressions (RE)
to be used instead of prefix expressions on edges of transducers.
Labeling of sub-expressions is not strictly tied to prefix expressions
and we could have also "labeled" versions of LTL or REs.
With these formalisms, however, we need to take care of things like ambiguity,
to make sure that matched labeled substrings are uniquely determined --
or, alternatively, allow non-unique matches and thus introduce more
non-determinism in matching. We chose to use prefix expressions as they
are deterministic and seem to be the right formalism for our needs,
but exploring possibilities of "labeled" REs or LTL could be interesting
future work. The advantage of using REs or LTL would be that these are well-known
formalisms, people know how to use them, and the specifications could be
shorter than with prefix expressions just because the non-determinism.
}

\subsubsection{Acknowledgements} This work was supported in part by the ERC-2020-AdG 101020093. 
The authors would like to thank Ana Oliveira da Costa for commenting on a draft of the paper.

\bibliographystyle{splncs04}
\bibliography{references}

\begin{thebibliography}{10}
\providecommand{\url}[1]{\texttt{#1}}
\providecommand{\urlprefix}{URL }
\providecommand{\doi}[1]{https://doi.org/#1}

\bibitem{Bartocci18}
Lectures on Runtime Verification - Introductory and Advanced Topics, LNCS, vol.
  10457. Springer (2018). \doi{10.1007/978-3-319-75632-5}

\bibitem{AcetoAAF22}
Aceto, L., Achilleos, A., Anastasiadi, E., Francalanza, A.: Monitoring
  hyperproperties with circuits. In: {FORTE} 2022. LNCS, vol. 13273, pp. 1--10.
  Springer (2022). \doi{10.1007/978-3-031-08679-3\_1}

\bibitem{AgrawalB16}
Agrawal, S., Bonakdarpour, B.: Runtime verification of $k$-safety
  hyperproperties in {HyperLTL}. In: {CSF} 2016. pp. 239--252. {IEEE} (2016).
  \doi{10.1109/CSF.2016.24}

\bibitem{AlurHK02}
Alur, R., Henzinger, T.A., Kupferman, O.: Alternating-time temporal logic. J.
  {ACM}  \textbf{49}(5),  672--713 (2002). \doi{10.1145/585265.585270}

\bibitem{Antimirov96}
Antimirov, V.M.: Partial derivatives of regular expressions and finite
  automaton constructions. Theor. Comput. Sci.  \textbf{155}(2),  291--319
  (1996). \doi{10.1016/0304-3975(95)00182-4}

\bibitem{Baeza-YatesG96}
Baeza{-}Yates, R.A., Gonnet, G.H.: Fast text searching for regular expressions
  or automaton searching on tries. J. {ACM}  \textbf{43}(6),  915--936 (1996).
  \doi{10.1145/235809.235810}

\bibitem{BartocciFHNC22}
Bartocci, E., Ferr{\`{e}}re, T., Henzinger, T.A., Nickovic, D., da~Costa, A.O.:
  Flavors of sequential information flow. In: {VMCAI} 2022. LNCS, vol. 13182,
  pp. 1--19. Springer (2022). \doi{10.1007/978-3-030-94583-1\_1}

\bibitem{bartocci2023hypernode}
Bartocci, E., Henzinger, T.A., Nickovic, D., da~Costa, A.O.: Hypernode automata
  (2023). \doi{10.48550/arXiv.2305.02836}

\bibitem{BaumeisterCBFS21}
Baumeister, J., Coenen, N., Bonakdarpour, B., Finkbeiner, B., S{\'{a}}nchez,
  C.: A temporal logic for asynchronous hyperproperties. In: {CAV} 2021. LNCS,
  vol. 12759, pp. 694--717. Springer (2021).
  \doi{10.1007/978-3-030-81685-8\_33}

\bibitem{BerglundM17}
Berglund, M., van~der Merwe, B.: Regular expressions with backreferences
  re-examined. In: Stringology Conference 2017. pp. 30--41. Czech Technical
  University in Prague (2017),
  \url{http://www.stringology.org/event/2017/p04.html}

\bibitem{BergstraBP94}
Bergstra, J.A., Bethke, I., Ponse, A.: Process algebra with iteration and
  nesting. Comput. J.  \textbf{37}(4),  243--258 (1994).
  \doi{10.1093/comjnl/37.4.243}, \url{https://doi.org/10.1093/comjnl/37.4.243}

\bibitem{BeutnerF21}
Beutner, R., Finkbeiner, B.: A temporal logic for strategic hyperproperties.
  In: {CONCUR} 2021. LIPIcs, vol.~203, pp. 24:1--24:19. Schloss Dagstuhl -
  Leibniz-Zentrum f{\"{u}}r Informatik (2021).
  \doi{10.4230/LIPIcs.CONCUR.2021.24}

\bibitem{BonakdarpourF18}
Bonakdarpour, B., Finkbeiner, B.: The complexity of monitoring hyperproperties.
  In: {CSF} 2018. pp. 162--174. {IEEE} (2018). \doi{10.1109/CSF.2018.00019}

\bibitem{BonakdarpourSS18}
Bonakdarpour, B., S{\'{a}}nchez, C., Schneider, G.: Monitoring hyperproperties
  by combining static analysis and runtime verification. In: {ISoLA} 2018.
  LNCS, vol. 11245, pp. 8--27. Springer (2018).
  \doi{10.1007/978-3-030-03421-4\_2}

\bibitem{BonakdarpourS21}
Bonakdarpour, B., Sheinvald, S.: Finite-word hyperlanguages. In: {LATA} 2021.
  LNCS, vol. 12638, pp. 173--186. Springer (2021).
  \doi{10.1007/978-3-030-68195-1\_17}

\bibitem{BozzelliMP15}
Bozzelli, L., Maubert, B., Pinchinat, S.: Unifying hyper and epistemic temporal
  logics. In: {FoSSaCS} 2015. LNCS, vol.~9034, pp. 167--182. Springer (2015).
  \doi{10.1007/978-3-662-46678-0\_11}

\bibitem{BozzelliPS21}
Bozzelli, L., Peron, A., S{\'{a}}nchez, C.: Asynchronous extensions of
  {HyperLTL}. In: {LICS} 2021. pp. 1--13. {IEEE} (2021).
  \doi{10.1109/LICS52264.2021.9470583}

\bibitem{BrettSB17}
Brett, N., Siddique, U., Bonakdarpour, B.: Rewriting-based runtime verification
  for alternation-free hyperltl. In: {TACAS} 2017. LNCS, vol. 10206, pp. 77--93
  (2017). \doi{10.1007/978-3-662-54580-5\_5}

\bibitem{Brzozowski64}
Brzozowski, J.A.: Derivatives of regular expressions. J. {ACM}  \textbf{11}(4),
   481--494 (1964). \doi{10.1145/321239.321249}

\bibitem{Bultan2017}
Bultan, T., Yu, F., Alkhalaf, M., Aydin, A.: Relational String Analysis, pp.
  57--68. Springer (2017). \doi{10.1007/978-3-319-68670-7_5}

\bibitem{extended}
Chalupa, M., Henzinger, T.A.: Monitoring hyperproperties with prefix
  transducers (2023), arXiv

\bibitem{ChalupaMLH23}
Chalupa, M., Muehlboeck, F., Lei, S.M., Henzinger, T.A.: Vamos: Middleware for
  best-effort third-party monitoring. In: {FASE} 2023. LNCS, vol. 13991, pp.
  260--281. Springer (2023). \doi{10.1007/978-3-031-30826-0\_15}

\bibitem{ChidaT22}
Chida, N., Terauchi, T.: On lookaheads in regular expressions with
  backreferences. In: {FSCD} 2022. LIPIcs, vol.~228, pp. 15:1--15:18. Schloss
  Dagstuhl - Leibniz-Zentrum f{\"{u}}r Informatik (2022).
  \doi{10.4230/LIPIcs.FSCD.2022.15}

\bibitem{ClarkeC97}
Clarke, C.L.A., Cormack, G.V.: On the use of regular expressions for searching
  text. {ACM} Trans. Program. Lang. Syst.  \textbf{19}(3),  413--426 (1997).
  \doi{10.1145/256167.256174}

\bibitem{ClarksonFKMRS14}
Clarkson, M.R., Finkbeiner, B., Koleini, M., Micinski, K.K., Rabe, M.N.,
  S{\'{a}}nchez, C.: Temporal logics for hyperproperties. In: {POST} 2014.
  LNCS, vol.~8414, pp. 265--284. Springer (2014).
  \doi{10.1007/978-3-642-54792-8\_15}

\bibitem{ClarksonS10}
Clarkson, M.R., Schneider, F.B.: Hyperproperties. J. Comput. Secur.
  \textbf{18}(6),  1157--1210 (2010). \doi{10.3233/JCS-2009-0393}

\bibitem{CoenenFHH19}
Coenen, N., Finkbeiner, B., Hahn, C., Hofmann, J.: The hierarchy of
  hyperlogics. In: {LICS} 2019. pp. 1--13. {IEEE} (2019).
  \doi{10.1109/LICS.2019.8785713}

\bibitem{CoenenFHHS21}
Coenen, N., Finkbeiner, B., Hahn, C., Hofmann, J., Schillo, Y.: Runtime
  enforcement of hyperproperties. In: {ATVA} 2021. LNCS, vol. 12971, pp.
  283--299. Springer (2021). \doi{10.1007/978-3-030-88885-5\_19}

\bibitem{Copi58}
Copi, I.M., Elgot, C.C., Wright, J.B.: Realization of events by logical nets.
  J. ACM  \textbf{5}(2),  181–196 (apr 1958). \doi{10.1145/320924.320931},
  \url{https://doi.org/10.1145/320924.320931}

\bibitem{DamanafshanF22}
Damanafshan, M., Fallah, M.S.: Monitorable hyperproperties of nonterminating
  systems. J. Log. Algebraic Methods Program.  \textbf{128},  100796 (2022).
  \doi{10.1016/j.jlamp.2022.100796}

\bibitem{DimitrovaFKRS12}
Dimitrova, R., Finkbeiner, B., Kov{\'{a}}cs, M., Rabe, M.N., Seidl, H.: Model
  checking information flow in reactive systems. In: {VMCAI} 2012. LNCS,
  vol.~7148, pp. 169--185. Springer (2012). \doi{10.1007/978-3-642-27940-9\_12}

\bibitem{EmersonH86}
Emerson, E.A., Halpern, J.Y.: "sometimes" and "not never" revisited: on
  branching versus linear time temporal logic. J. {ACM}  \textbf{33}(1),
  151--178 (1986). \doi{10.1145/4904.4999}

\bibitem{FalconeMFR11}
Falcone, Y., Mounier, L., Fernandez, J., Richier, J.: Runtime enforcement
  monitors: composition, synthesis, and enforcement abilities. Formal Methods
  Syst. Des.  \textbf{38}(3),  223--262 (2011). \doi{10.1007/s10703-011-0114-4}

\bibitem{FinkbeinerHT19}
Finkbeiner, B., Haas, L., Torfah, H.: Canonical representations of k-safety
  hyperproperties. In: {CSF} 2019. pp. 17--31. {IEEE} (2019).
  \doi{10.1109/CSF.2019.00009}

\bibitem{FinkbeinerHST17}
Finkbeiner, B., Hahn, C., Stenger, M., Tentrup, L.: Monitoring hyperproperties.
  In: {RV} 2017. LNCS, vol. 10548, pp. 190--207. Springer (2017).
  \doi{10.1007/978-3-319-67531-2\_12}

\bibitem{FinkbeinerHST18}
Finkbeiner, B., Hahn, C., Stenger, M., Tentrup, L.: {RVHyper}: {A} runtime
  verification tool for temporal hyperproperties. In: {TACAS} 2018. LNCS, vol.
  10806, pp. 194--200. Springer (2018). \doi{10.1007/978-3-319-89963-3\_11}

\bibitem{FinkbeinerHST19}
Finkbeiner, B., Hahn, C., Stenger, M., Tentrup, L.: Monitoring hyperproperties.
  Formal Methods Syst. Des.  \textbf{54}(3),  336--363 (2019).
  \doi{10.1007/s10703-019-00334-z}

\bibitem{FinkbeinerHST20}
Finkbeiner, B., Hahn, C., Stenger, M., Tentrup, L.: Efficient monitoring of
  hyperproperties using prefix trees. Int. J. Softw. Tools Technol. Transf.
  \textbf{22}(6),  729--740 (2020). \doi{10.1007/s10009-020-00552-5}

\bibitem{Finkbeiner017}
Finkbeiner, B., Zimmermann, M.: The first-order logic of hyperproperties. In:
  {STACS} 2017. LIPIcs, vol.~66, pp. 30:1--30:14. Schloss Dagstuhl -
  Leibniz-Zentrum f{\"{u}}r Informatik (2017).
  \doi{10.4230/LIPIcs.STACS.2017.30}

\bibitem{GutsfeldMO20}
Gutsfeld, J.O., M{\"{u}}ller{-}Olm, M., Ohrem, C.: Propositional dynamic logic
  for hyperproperties. In: {CONCUR} 2020. LIPIcs, vol.~171, pp. 50:1--50:22.
  Schloss Dagstuhl - Leibniz-Zentrum f{\"{u}}r Informatik (2020).
  \doi{10.4230/LIPIcs.CONCUR.2020.50}

\bibitem{GutsfeldMO21}
Gutsfeld, J.O., M{\"{u}}ller{-}Olm, M., Ohrem, C.: Automata and fixpoints for
  asynchronous hyperproperties. In: {POPL} 2021. pp. 1--29 (2021).
  \doi{10.1145/3434319}

\bibitem{Hahn19}
Hahn, C.: Algorithms for monitoring hyperproperties. In: {RV} 2019. LNCS, vol.
  11757, pp. 70--90. Springer (2019). \doi{10.1007/978-3-030-32079-9\_5}

\bibitem{HahnST19}
Hahn, C., Stenger, M., Tentrup, L.: Constraint-based monitoring of
  hyperproperties. In: {TACAS} 2019. LNCS, vol. 11428, pp. 115--131. Springer
  (2019). \doi{10.1007/978-3-030-17465-1\_7}

\bibitem{HalleK17}
Hall{\'{e}}, S., Khoury, R.: Event stream processing with beepbeep 3. In:
  {RV-CuBES} 2017. Kalpa Publications in Computing, vol.~3, pp. 81--88.
  EasyChair (2017). \doi{10.29007/4cth}

\bibitem{HanWW05}
Han, Y., Wang, Y., Wood, D.: Prefix-free regular-expression matching. In: {CPM}
  2005. LNCS, vol.~3537, pp. 298--309. Springer (2005).
  \doi{10.1007/11496656\_26}

\bibitem{HanW04}
Han, Y., Wood, D.: The generalization of generalized automata: Expression
  automata. In: {CIAA} 2004, Revised Selected Papers. LNCS, vol.~3317, pp.
  156--166. Springer (2004). \doi{10.1007/978-3-540-30500-2\_15}

\bibitem{HopcroftU79}
Hopcroft, J.E., Ullman, J.D.: Introduction to Automata Theory, Languages and
  Computation. Addison-Wesley (1979)

\bibitem{HuismanWS06}
Huisman, M., Worah, P., Sunesen, K.: A temporal logic characterisation of
  observational determinism. In: {CSFW} 2006. p.~3. {IEEE} (2006).
  \doi{10.1109/CSFW.2006.6}

\bibitem{Kleene}
Kleene, S.C.: Representation of events in nerve nets and finite automata.
  Automata Studies pp. 3--41 (1956)

\bibitem{KrebsMV018}
Krebs, A., Meier, A., Virtema, J., Zimmermann, M.: Team semantics for the
  specification and verification of hyperproperties. In: {MFCS} 2018. LIPIcs,
  vol.~117, pp. 10:1--10:16. Schloss Dagstuhl - Leibniz-Zentrum f{\"{u}}r
  Informatik (2018). \doi{10.4230/LIPIcs.MFCS.2018.10}

\bibitem{LeuckerSS0S18}
Leucker, M., S{\'{a}}nchez, C., Scheffel, T., Schmitz, M., Schramm, A.: Tessla:
  runtime verification of non-synchronized real-time streams. In: {SAC} 2018.
  pp. 1925--1933. {ACM} (2018). \doi{10.1145/3167132.3167338}

\bibitem{McLean90}
McLean, J.: Security models and information flow. In: {SP} 1990. pp. 180--189.
  {IEEE} (1990). \doi{10.1109/RISP.1990.63849}

\bibitem{PennaITZ03}
Penna, G.D., Intrigila, B., Tronci, E., Zilli, M.V.: Synchronized regular
  expressions. Acta Informatica  \textbf{39}(1),  31--70 (2003).
  \doi{10.1007/s00236-002-0099-y}

\bibitem{PinisettySS18}
Pinisetty, S., Schneider, G., Sands, D.: Runtime verification of
  hyperproperties for deterministic programs. In: {FormaliSE} 2018. pp. 20--29.
  {ACM} (2018). \doi{10.1145/3193992.3193995}

\bibitem{Pnueli77}
Pnueli, A.: The temporal logic of programs. In: {FOCS} 1977. pp. 46--57. {IEEE}
  (1977). \doi{10.1109/SFCS.1977.32}

\bibitem{Raabe16}
Raabe, M.N.: A Temporal Logic Approach to Information-flow Control. Ph.D.
  thesis, Saarland University (2016)

\bibitem{Rutten98}
Rutten, J.J.M.M.: Automata and coinduction (an exercise in coalgebra). In:
  {CONCUR} '98. LNCS, vol.~1466, pp. 194--218. Springer (1998).
  \doi{10.1007/BFb0055624}

\bibitem{Sanchez21}
S{\'{a}}nchez, C.: Synchronous and asynchronous stream runtime verification.
  In: {VORTEX} 2021. pp.~5--7. {ACM} (2021). \doi{10.1145/3464974.3468453}

\bibitem{Schneider00}
Schneider, F.B.: Enforceable security policies. {ACM} Trans. Inf. Syst. Secur.
  \textbf{3}(1),  30--50 (2000). \doi{10.1145/353323.353382}

\bibitem{Sistla83}
Sistla, A.P.: Theoretical Issues in the Design and Verification of Distributed
  Systems. Ph.D. thesis, Harvard University (1983)

\bibitem{StuckiSSB19}
Stucki, S., S{\'{a}}nchez, C., Schneider, G., Bonakdarpour, B.: Gray-box
  monitoring of hyperproperties. In: {FM} 2019. LNCS, vol. 11800, pp. 406--424.
  Springer (2019). \doi{10.1007/978-3-030-30942-8\_25}

\bibitem{VeanesHLMB12}
Veanes, M., Hooimeijer, P., Livshits, B., Molnar, D., Bj{\o}rner, N.S.:
  Symbolic finite state transducers: algorithms and applications. In: {POPL}
  2012. pp. 137--150. {ACM} (2012). \doi{10.1145/2103656.2103674}

\bibitem{VirtemaHFK021}
Virtema, J., Hofmann, J., Finkbeiner, B., Kontinen, J., Yang, F.: Linear-time
  temporal logic with team semantics: Expressivity and complexity. In: {FSTTCS}
  2021. LIPIcs, vol.~213, pp. 52:1--52:17. Schloss Dagstuhl - Leibniz-Zentrum
  f{\"{u}}r Informatik (2021). \doi{10.4230/LIPIcs.FSTTCS.2021.52}

\bibitem{ZdancewicM03}
Zdancewic, S., Myers, A.C.: Observational determinism for concurrent program
  security. In: {CSFW} 2003. p.~29. {IEEE} (2003).
  \doi{10.1109/CSFW.2003.1212703}

\end{thebibliography}


\end{document}